\documentclass[11pt,a4paper]{article}
\usepackage[english]{babel} %
\usepackage[T1]{fontenc}
\usepackage{graphicx,xspace,amsmath,enumitem}
\usepackage{wrapfig,dsfont,subfig}
\usepackage[colorlinks=true,urlcolor=blue,citecolor=red]{hyperref}
\usepackage{booktabs}
\usepackage[nofancy]{svninfo}
\usepackage{xcolor}

\usepackage[boldsans]{concmath} %
\usepackage{euler} %

\usepackage{amsfonts,amssymb,bm,amsthm,fullpage,thm-restate}

\newtheorem{lemma}[equation]{Lemma}

\newtheorem{theorem}[equation]{Theorem}

\newcommand{\minipar}[1]{\smallskip\textbf{#1}} %
\newcommand{\N}{\ensuremath{\mathds N}} %
\newcommand{\R}{\ensuremath{\mathds R}} %
\newcommand{\figscale}{1.1}

\graphicspath{{figures/}}

\begin{document}

\title{Arc diagrams, flip distances, and Hamiltonian triangulations\footnote{An
    extended abstract~\cite{chktw-adfdht-15} of this work was presented at the
    32nd Symposium on Theoretical Aspects of Computer Science (STACS~2015).}}

\author{%
  Jean Cardinal\thanks{Universit\'e libre de Bruxelles (ULB), Belgium,
    \texttt{jcardin@ulb.ac.be}. Partially supported by the ESF EUROCORES
    programme EuroGIGA, CRP ComPoSe.} %
  \and %
  Michael Hoffmann\thanks{Department of Computer Science, ETH Z\"urich,
    Switzerland, \texttt{\{hoffmann,kustersv,manuelwe\}@inf.ethz.ch}.}
  \thanks{Partially supported by the ESF EUROCORES programme EuroGIGA, CRP GraDR
    and the Swiss National Science Foundation, SNF Project 20GG21-134306.} %
  \and %
  Vincent Kusters\footnotemark[3] \footnotemark[4] %
  \and %
  Csaba D. T\'oth\thanks{California State University Northridge, Los Angeles,
    CA; and Tufts University, Medford, MA, USA, \texttt{cdtoth@acm.org}. 
    Partially supported by the NSF awards CCF-1422311 and CCF-1423615.}%
  \and %
  Manuel Wettstein\footnotemark[3] \thanks{Partially supported by the ESF
    EUROCORES programme EuroGIGA, CRP ComPoSe and the Swiss National Science
    Foundation, SNF Project 20GG21-134318/1.} %
}

\maketitle

\begin{abstract}
  We show that every triangulation (maximal planar graph) on $n\ge 6$ vertices
  can be flipped into a Hamiltonian triangulation using a sequence of less than
  $n/2$ combinatorial edge flips. The previously best upper bound uses
  $4$-connectivity as a means to establish Hamiltonicity. But in general about
  $3n/5$ flips are necessary to reach a $4$-connected triangulation. Our result
  improves the upper bound on the diameter of the flip graph of combinatorial
  triangulations on $n$ vertices from $5.2n-33.6$ to $5n-23$. We also show that
  for every triangulation on $n$ vertices there is a simultaneous flip of less
  than $2n/3$ edges to a $4$-connected triangulation. The bound on the number of
  edges is tight, up to an additive constant. As another application we show
  that every planar graph on $n$ vertices admits an arc diagram with less than
  $n/2$ biarcs, that is, after subdividing less than $n/2$ (of potentially
  $3n-6$) edges the resulting graph admits a $2$-page book embedding.
\end{abstract}

\section{Introduction}\label{sec:intro}

An \emph{arc diagram} (\figurename~\ref{fig:0}) is a drawing of a graph in which
vertices are represented by points on a horizontal line, called the
\emph{spine}, and edges are drawn either as one halfcircle (\emph{proper arc})
or as a sequence of halfcircles centered on the line (forming a smooth Jordan
arc). In a \emph{proper arc diagram} all arcs are proper.  Arc diagrams have
been used and studied in many contexts since their first appearance in the
mid-sixties~\cite{n-ppmnc-68,s-mnicg-64}. They constitute a well-studied
geometric representation in graph drawing~\cite{ddl-srd-13} that occurs, for
instance, in the study of crossing numbers~\cite{aafrs-sdccn-14,bk-mpccp-64} and
universal point sets for circular arc drawings~\cite{aefklmtw-upsdpgca-14}.

\begin{figure}[hbtp]
  \centering%
  \subfloat[]{\includegraphics{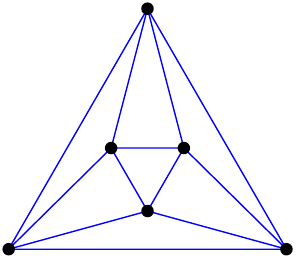}}\hfil
  \subfloat[]{\includegraphics{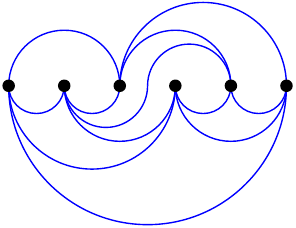}\label{fig:1:b}}\hfil
  \subfloat[]{\includegraphics{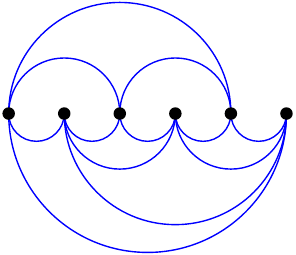}\label{fig:1:c}}\hfil
  \caption{A plane straight-line drawing (a), an arc diagram (b) and a
    proper arc-diagram (c) of the same graph.\label{fig:0}}
\end{figure}

Bernhart and Kainen~\cite{bk-btg-79} proved that a planar graph admits a
\emph{plane} (i.e., crossing-free) proper arc diagram if and only if it can be augmented
to a Hamiltonian planar graph by adding new edges. Such planar graphs are also
called \emph{subhamiltonian}, and they are NP-hard to recognize~\cite{w-chcpm-82}. A
Hamiltonian cycle in the augmented graph directly yields a feasible order for
the vertices on the spine. Every planar graph can be subdivided into a
subhamiltonian graph with at most one subdivision vertex per
edge~\cite{pw-epgfv-01}. Consequently, every planar graph admits a plane
\emph{biarc diagram} in which each edge is either a proper arc or the union of
two halfcircles (a \emph{biarc}); one above and one below the spine. Di Giacomo
et al.~\cite{ddlw-ccdpg-05} showed that every planar graph even admits a
\emph{monotone} plane biarc diagram in which every biarc is $x$-monotone---such an
embedding is also called a \emph{2-page topological book embedding}.
See~\cite{ddl-srd-13} for various other applications of subhamiltonian
subdivisions of planar graphs.

Eppstein~\cite{e-upsdpgca-14} said: ``Arc diagrams (with one arc per edge) are
very usable and practical but can only handle a subset of planar graphs.'' Using
biarcs allows us to represent all planar graphs, but adds to the complexity of the
drawing. Hence it is a natural question to ask: How close can we get to a proper
arc diagram, while still being able to represent all planar graphs?  A natural
measure of complexity is the number of biarcs used.

Previous methods for subdividing an $n$-vertex planar graph into a
subhamiltonian graph use at most one subdivision per
edge~\cite{ddl-srd-13,ddlw-ccdpg-05,kw-evpfb-02,pw-epgfv-01}, consequently the
number of biarcs in an arc diagram is bounded by the number of edges. Our main
goal in this paper is to tighten the upper and lower bounds on the minimum
number of biarcs in an arc diagram (or, alternatively, the number of
subdivision vertices in a subhamiltonian subdivision) of a planar graph with
$n$ vertices. Minimizing the number of biarcs is clearly NP-hard, since the
number of biarcs is zero if and only if the graph is subhamiltonian.

\minipar{Our results.}  In Section \ref{sec:monotone} we show that the number of
biarcs can be bounded by $n$, even when they are restricted to be monotone. Although
previous methods can be shown to yield less than the trivial $3n-6$
biarcs~\cite{kw-evpfb-02}, or ensure monotonicity~\cite{ddlw-ccdpg-05}, we
give the first proof that both properties can be guaranteed simultaneously. The
algorithm is similar to the canonical ordering-based method of Di Giacomo et
al.~\cite{ddlw-ccdpg-05}.

\begin{restatable}{theorem}{biarcmonotone}\label{thm:biarcmonotone}
  Every planar graph on $n\ge 4$ vertices admits a plane biarc diagram using at
  most $n-4$ biarcs, all of which are monotone. Moreover, such a diagram can be
  computed in $O(n)$ time.
\end{restatable}

For arbitrary (not necessarily monotone) biarcs we achieve better bounds.  Our
main tool is relating subhamiltonian planar graphs to edge flips and subdivisions in
triangulations.

A \emph{flip} in a triangulation involves switching the diagonal of a
quadrilateral made of two adjacent facial triangles. We consider
\emph{combinatorial} flips, which can be regarded as an operation on an abstract
graph. The \emph{flip graph} induced by flips on the set of all triangulations
on $n$ vertices, and the corresponding \emph{flip distance} between two
triangulations, have been the topic of extensive
research~\cite{bh-fpg-09,bv-hfct-12}. For instance, the flip diameter restricted
to the interior of a convex polygon is equivalent to the rotation distance of binary
trees~\cite{p-da-14,stt-rdthg-88}.

By \emph{subdividing} an edge $e$ we mean replacing $e$ with a new vertex
that is connected to both endpoints of $e$. The following theorem, which is proved
in Section \ref{sec:generalbiarcs}, relates biarcs to edge subdivisions and is
a simple generalization of the characterization of Bernhart and Kainen.

\begin{restatable}{theorem}{bernkaingeneral}\label{thm:bernkaingeneral}
  A planar graph $G$ admits a plane biarc diagram with at most $k$ biarcs if and
  only if there is a set of at most $k$ edges in $G$ so that subdividing these
  edges transforms $G$ into a subhamiltonian graph.
\end{restatable}

In Section \ref{sec:simflip} we prove that in every triangulation there exists a
set of less than $2n/3$ edges that can be flipped \emph{simultaneously} so that
the resulting triangulation is $4$-connected, and that this bound is tight up to
an additive constant. Since by Tutte's Theorem every
$4$-connected planar graph is Hamiltonian, we can transform every planar graph
into a subhamiltonian graph by subdividing at most $2n/3$ edges. The fact that a
single simultaneous flip can make a triangulation $4$-connected has already been
established by Bose et al.~\cite{bcgmw-sdfpt-07}. However, they do not give any
bound on the number of flipped edges.

\begin{restatable}{theorem}{upsimfour}\label{thm:upsim4}
  Every maximal planar graph on $n\ge 6$ vertices can be transformed into a
  $4$-connected maximal planar graph using a simultaneous flip of at most
  $\lfloor(2n-7)/3\rfloor$ edges. Moreover, such a set of simultaneously
  flippable edges can be computed in $O(n^2)$ time.
\end{restatable}

\begin{restatable}{theorem}{upsimfourlower}\label{thm:upsim4lower}
  For every $i\in\N$, there is a maximal planar graph $G_i$ on $n_i=3i+4$
  vertices such that no simultaneous flip of less than $(2n_i-8)/3=2i$ edges
  results in a $4$-connected graph.
\end{restatable}

Finally, in Section \ref{sec:hamflip} we prove an upper bound on the flip
distance of a triangulation to Hamiltonicity, that is, on the worst-case number
of successive flips required to reach a Hamiltonian triangulation.  Given the
hardness of determining whether a given planar graph is Hamiltonian, we should
not expect a nice characterization of (non-)Hamiltonicity. Hence, in the context
of planar graphs, $4$-connectivity is often used as a substitute because by
Tutte's Theorem it is a sufficient condition for Hamiltonicity.

Bose et al.~\cite{bjrsv-mt4uf-14} gave a tight bound (up to an additive
constant) of $3n/5$ on the number of flips that transform a given triangulation
on $n$ vertices into a $4$-connected triangulation. We show that fewer flips are
sufficient to guarantee Hamiltonicity. Obviously, the target triangulation is
not $4$-connected in general, which means it possibly contains separating
triangles.

\begin{restatable}{theorem}{hamflip}\label{thm:hamflip}
  Every maximal planar graph on $n\ge 6$ vertices can be transformed into a
  Hamiltonian maximal planar graph using a sequence of at most
  $\lfloor(n-3)/2\rfloor$ edge flips. Alternatively, it can be transformed into
  a subhamiltonian planar graph by subdividing a set of at most
  $\lfloor(n-3)/2\rfloor$ edges. Moreover, such a sequence of flips or
  subdivisions can be computed in $O(n^2)$ time.
\end{restatable}

In this case we do not have a matching lower bound. The best lower bound we know
can be obtained using \emph{Kleetopes}~\cite{g-cp-03}. These are convex
polytopes that are generated from another convex polytope by replacing every
face by a small pyramid. In the language of planar graphs, we start from a
$3$-connected planar graph and for every face add a new vertex that is connected
to all vertices on the boundary of the face. If the graph we start from has
enough faces, then the added vertices form a large independent set so that the
resulting graph is not Hamiltonian. Aichholzer~et~al.~\cite{ahk-tpst-08}
describe such a construction explicitly in the context of flipping a
triangulation to a Hamiltonian triangulation, but state the asymptotics only.
A precise counting reveals the following figures.
\begin{restatable}{theorem}{hamlower}\label{thm:hamlower}
  For every $i\in\N$, there is a maximal planar graph $G_i$ on $n_i=3i+8$
  vertices such that no sequence of less than $(n_i-8)/3=i$ edge flips produces
  a Hamiltonian graph, and there is no set of less than $(n_i-8)/3=i$ edges
  whose subdivision produces a subhamiltonian graph.
\end{restatable}

Our proof of Theorem~\ref{thm:hamflip} is constructive, and each flip in the
sequence involves an edge of the initial graph $G$ that is incident to a
separating triangle of $G$. Several of these edges may be incident to a common
facial triangle, in which case the edges are not simultaneously flippable.

Theorem~\ref{thm:bernkaingeneral} allows us to translate Theorems
\ref{thm:hamflip} and \ref{thm:hamlower} to the context of biarc diagrams, where
we obtain bounds for the number of biarcs needed.

\begin{restatable}{corollary}{arcdiaupper}\label{cor:arcdiaupper}
  Every planar graph on $n \ge 6$ vertices admits a plane biarc diagram with at
  most $\lfloor(n-3)/2\rfloor$ biarcs. Moreover, such a diagram can be computed
  in $O(n^2)$ time.
\end{restatable}

\begin{restatable}{corollary}{arcdialower}\label{cor:arcdialower}
  For every $i\in\N$, there is a maximal planar graph $G_i$ on $n_i=3i+8$
  vertices that cannot be drawn as a plane biarc diagram using less than
  $(n_i-8)/3=i$ biarcs.
\end{restatable}

As another corollary, we establish a new upper bound on the diameter of the flip
graph of all triangulations on $n$ vertices, improving on the previous best
bound of $5.2n-33.6$ by Bose et al.~\cite{bjrsv-mt4uf-14}. Mori et
al.~\cite{mno-dfhts-03} showed that any two Hamiltonian triangulations on $n$
vertices can be transformed into each other by a sequence of at most
$\max\{4n-20,0\}$ flips. Combined with Theorem~\ref{thm:hamflip}, this implies
the following.
\begin{restatable}{corollary}{flipdist}\label{thm:flipdist}
  Every two triangulations on $n\ge 6$ vertices can be transformed into each
  other using a sequence of at most $5n-23$ edge flips.
\end{restatable}

\section{Notation}\label{sec:notation}

A \emph{drawing} of a graph $G$ in $\R^2$ maps the vertices into distinct points
in the plane and maps each edge to a Jordan arc between (the images of) the two
vertices that is disjoint from (the image of) any other vertex. To avoid
notational clutter it is common to identify vertices and edges with their
geometric representation. A drawing is called \emph{plane} (or an
\emph{embedding}) if no two edges intersect except at a possible common
endpoint. Only planar graphs admit plane drawings, but not every drawing of a
planar graph is plane. A \emph{maximal planar} graph on $n$ vertices is a
planar graph with $3n-6$ edges. In this paper the term \emph{triangulation} is
used as a synonym for maximal planar graph.\footnote{In contrast, a maximal
  plane \emph{straight-line} drawing may have fewer edges, depending on the
  number of points on the convex hull.}

In a plane drawing of a triangulation $G$, every face (including the outer
face) is bounded by three edges. Hence, every triangulation with $n\geq 4$
vertices is 3-connected~\cite{Die16}[Lemma~4.4.5]. Every 3-connected planar
graph has a topologically unique plane drawing, apart from the choice of the
outer face. Specifically, the facial triangles are precisely the nonseparating
chordless cycles of $G$ in every plane drawing~\cite{Die16}[Proposition~4.2.7].
Consequently, $G$ has a well-defined dual graph $G^*$ (independent of the
drawing): the vertices of $G^*$ correspond to the faces of $G$, and two
vertices of $G^*$ are adjacent if and only if the corresponding faces share an
edge. A triangle of $G$ that is not facial is called a \emph{separating}
triangle, as its removal disconnects the graph.

A graph is \emph{Hamiltonian} if it contains a cycle through all vertices. By a
famous theorem of Tutte~\cite{Tut56,Whi31}, all $4$-connected planar graphs are Hamiltonian. For
triangulations, $4$-connectivity is equivalent to the absence of separating
triangles. 
A vertex or an edge is \emph{incident} to a triangle $T$ in a graph if it is a
vertex or edge of $T$.

A triangulation $G$ can be partitioned into a \emph{$4$-block tree}
$\mathcal{B}$. Each vertex of $\mathcal{B}$ is either a maximal $4$-connected
component of $G$ or a subgraph of $G$ that is isomorphic to $K_4$. Two vertices
of $\mathcal{B}$ are adjacent if they share a separating triangle of $G$. The
4-block tree is similar to the standard block-tree for $2$-connected components,
but the generalization of the notion ``component'' to higher connectivity is not
straightforward in general. 
For a triangulation, however, the $4$-block tree is well-defined and can be
computed in linear time and space~\cite{k-amcvr-97}.

\minipar{Flips.} Consider an edge $ab$ of a triangulation $G$ and let
$abc$ and $adb$ denote the two incident facial triangles. The \emph{flip} of
$ab$ replaces the edge $ab$ by the edge $cd$. If this operation produces a
triangulation (i.e., if $c\neq d$ and the edge $cd$ is not already present in $G$), we call
$ab$ \emph{flippable}\footnote{We consider \emph{combinatorial flips}, as
  opposed to \emph{geometric flips} defined for straight-line plane drawings,
  where an edge is flippable if and only if the quadrilateral formed by the two
  incident facial triangles is convex.}.

A closely related concept is the \emph{simultaneous flip} of a set $F$ of
flippable edges in a triangulation $G=(V,E)$, which is defined as follows. For
$e\in F$ denote by $\mathrm{c}(e)$ the edge created by flipping $e$ in $G$, and
let $\mathrm{C}(F)=\bigcup_{e\in F}\mathrm{c}(e)$. Then the simultaneous flip of
$F$ in $G$ results in the graph $G'=(V,(E\setminus F)\cup\mathrm{C}(F))$. Bose
et al.~\cite{bcgmw-sdfpt-07} introduced this notion and showed that the result
of a simultaneous flip is a triangulation if every facial triangle of $G$ is
incident to at most one edge from $F$ and the edges $c(e)$, for $e\in F$, are
all distinct and not present in $E$.

\section{Monotone Biarc Diagrams}\label{sec:monotone}

In this section we present a simple linear time algorithm to construct a biarc
diagram in which all biarcs are drawn as monotone curves (with respect to the
spine). The algorithm is based on the fundamental notion of a canonical
ordering, which is defined for an \emph{embedded} triangulation. As every
triangulation on $n\ge 4$ vertices is $3$-connected, embedding it into the plane
essentially amounts to selecting one facial triangle to be the \emph{outer
  face}. This choice also determines a unique outer face (cycle) for every
biconnected subgraph.

A \emph{canonical ordering}~\cite{fpp-hdpgg-90} for an embedded triangulation
$G$ on $n$ vertices is a total order of the vertices $v_1,\ldots,v_n$ such that
\begin{itemize}
\item for $i\in\{3,\ldots,n\}$, the induced subgraph $G_i=G[\{v_1,\ldots,
  v_i\}]$ is biconnected and internally triangulated (i.e., every face other
  than the outer face $C_i$ is a triangle);
\item for $i\in\{3,\ldots,n\}$, $v_1v_2$ is an edge of $C_i$;
\item for $i\in\{3,\ldots,n-1\}$, $v_{i+1}$ lies in the interior of $C_i$
  (the unbounded region of the plane bounded by $C_i$) and the neighbors of
  $v_{i+1}$ in $G_i$ form a sequence of consecutive vertices along the boundary of $C_i$.
\end{itemize}
It is well-known that every triangulation admits a canonical
ordering~\cite{fpp-hdpgg-90}, and such an ordering can be computed in $O(n)$
time~\cite{cp-ltadp-95}.

\biarcmonotone*
\begin{proof}
  Let $G$ be a planar graph on $n\ge 4$ vertices and suppose without loss of
  generality that $G$ is an embedded triangulation. If $G$ is not maximal
  planar, add edges to make $G$ maximal planar, choose any embedding,
  and simply remove the added edges from the final drawing.

  Consider a canonical ordering $v_1,\ldots,v_n$ of the vertices of $G$.  We
  construct a biarc diagram of $G$ incrementally by inserting the vertices in
  canonical order and embedding them on the $x$-axis (spine). Let
  $G_i=G[\{v_1,\ldots,v_i\}]$ and let $C_i$ denote the outer cycle of $G_i$, for
  $i=3,\ldots , n$.  During the algorithm, we maintain the following invariants.
  \begin{enumerate}[label={(I\arabic*)}]\setlength{\itemindent}{2\labelsep}
  \item\label{i:first} All edges of $C_i$ are proper arcs (none is a biarc). The
    vertices $v_1$ and $v_2$ are the leftmost and rightmost, respectively,
    vertices of $G_i$ on the spine. The edge $v_1v_2$ forms the lower envelope
    of $G_i$ (i.e., no point of the biarc diagram is vertically below). All edges of
    $C_i$ other than $v_1v_2$ are on the upper envelope of $G_i$ (i.e., no point
    of the biarc diagram  is vertically above).
  \item\label{i:last} Any biarc used in $G_i$ is a down-up biarc, that is, the
    semicircle incident to its left endpoint lies below the spine and the
    semicircle incident to its right endpoint lies above the spine.
  \end{enumerate}

  We embed the triangle $G_3$ by placing $v_1$, $v_3$, and $v_2$ on the spine in
  this order from left to right and by drawing all edges as proper arcs below the spine
  (\figurename~\ref{fig:1}). Clearly \ref{i:first}--\ref{i:last} hold for this
  embedding.
  \begin{figure}[htbp]
    \centering%
    \includegraphics[scale=\figscale]{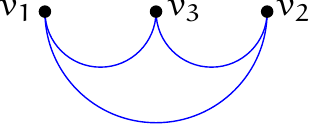}
    \caption{Start the incremental embedding with a triangle.\label{fig:1}}
  \end{figure}

  Now suppose that we have a biarc diagram for $G_i$ that satisfies the invariants
  and we want to add $v_{i+1}$. Let $w_1,\ldots,w_{k_i}$ be the vertices of
  $C_i$ labeled from left to right along the spine. By \ref{i:first} this order
  is compatible with the vertex order along $C_i$, with $v_1=w_1$ and
  $v_2=w_{k_i}$. As we work with a canonical ordering, the neighbors of
  $v_{i+1}$ on $C_i$ form a contiguous subsequence $w_{\ell_i},\ldots,w_{r_i}$
  of $C_i$, with $1\le\ell_i<r_i\le k_i$. In addition, \ref{i:last} guarantees
  that we can insert $v_{i+1}$ along the spine between $w_{\ell_i}$ and
  $w_{\ell_i+1}$, just to the right of $w_{\ell_i}$: Every biarc leaving
  $w_{\ell_i}$ to the right goes down first and, therefore, does not block the
  spine locally at $w_{\ell_i}$, whereas proper arcs above the spine leaving
  $w_{\ell_i}$ to the right can be bent down to become down-up biarcs while
  maintaining their vertical order (\figurename~\ref{fig:2}). After placing
  $v_{i+1}$, the edges to $w_{\ell_i},\ldots,w_{r_i}$ can be drawn as proper
  arcs above the spine. The edge $w_{\ell_i}v_{i+1}$ can even be drawn as a
  proper arc below the spine because the two vertices are neighbors along the
  spine by construction. It is easily checked that the invariants
  \ref{i:first}--\ref{i:last} are maintained. This completes the description of
  the first version of our algorithm.
  \begin{figure}[htbp]
    \centering%
    \subfloat[before]{\includegraphics[scale=\figscale]{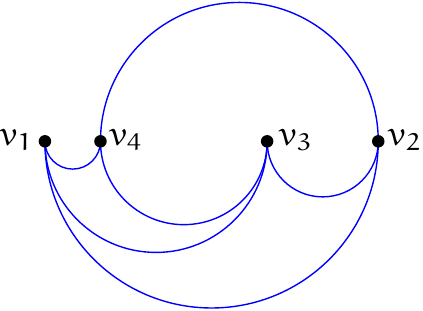}}\hfil
    \subfloat[after]{\includegraphics[scale=\figscale]{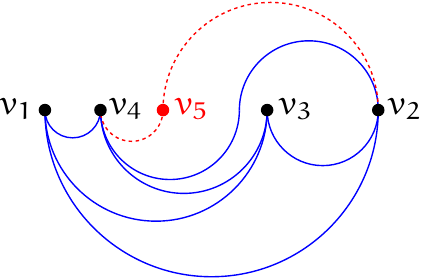}}\hfil
    \caption{Make room for a new vertex $v_5$.\label{fig:2}}
  \end{figure}

  \paragraph{First lower bound on the number of proper arcs.}  It remains
  to bound the number of biarcs used by the algorithm. As a first
  observation, note that all edges are drawn as proper arcs initially (when
  they first appear). An edge $e$ may become a biarc in a later step
  only if it is bent down to make room for a vertex inserted
  immediately to the right of the left endpoint of $e$. In particular,
  every edge drawn below the spine, such as the three edges of $G_3$
  and the edges $w_{\ell_i}v_{i+1}$ drawn at steps
  $i\in\{3,\ldots,n-2\}$, remain proper arcs throughout the algorithm.
  Finally, at least three new edges will be drawn in the last step $i=n-1$
  (i.e., when inserting $v_n$) as proper arcs. This yields a
  first lower bound of at least $3 + (n-4) + 3 = n+2$ proper arcs and, therefore,
  at most $3n-6-(n+2)=2n-8$ biarcs.

  \paragraph{A refined algorithm and lower bound.}  In order to obtain the
  claimed bound, let us consider in more detail the insertion of a vertex
  $v_{i+1}$ where $i \in \{3,\dots,n-1\}$. We claim that for any vertex
  $v_{i+1}$ we can obtain $r_i-\ell_i$ proper arcs in the final drawing, rather
  than just one. However, we also have to adapt our algorithm slightly, as
  described in the following paragraph.

  The improvement is based on two simple but crucial observations.  First,
  observe that a vertex $v_{i+1}$ can be inserted just to the right of any of
  the vertices $w_{\ell_i},\ldots,w_{r_i-1}$, not only $w_{\ell_i}$.  The
  invariants \ref{i:first}--\ref{i:last} can be maintained for any such choice.
  Second, observe that none of the edges of the path
  $w_{\ell_i+1},\ldots,w_{r_i}$ appear on $C_{i+1}$ anymore and neither do the
  left endpoints of these edges.  In particular, it follows that every proper
  arc among those edges will remain a proper arc in the final drawing. Now we
  have to be careful when counting these edges because some of them might be
  drawn below the spine and we accounted for them already. Here is where the
  first observation comes to our help. We modify the algorithm to insert
  $v_{i+1}$ just to the right of the last vertex $w_{f_i}$ in
  $w_{\ell_i},\ldots,w_{r_i-1}$ such that the edge $w_{f_i}w_{f_i+1}$ is drawn
  below the spine. If no such edge exists, then we insert $v_{i+1}$ just to the
  right of $w_{\ell_i}$, as before.

  For the analysis we consider two cases. If $w_{f_i}$ exists
  (\figurename~\ref{fig:wfiexists}), then the insertion of $v_{i+1}$ does not
  create any biarcs. All edges along the path $w_{f_i+1},\ldots,w_{r_i}$ are
  proper arcs drawn above the spine and have not yet been counted. As none of
  these arcs appears on $C_{i+1}$, they will not be counted again. In addition,
  all edges from $v_{i+1}$ to $w_{\ell_i+1},\ldots,w_{f_i}$ are proper arcs of
  $G_{i+1}$ whose left endpoints do not appear on $C_{i+1}$. Lastly, the edge
  from $v_{i+1}$ to $w_{f_i+1}$ can be drawn below the spine. Therefore, all
  these edges remain proper arcs throughout the algorithm. The total number of
  new proper arcs in the final drawing is, therefore, at least
  $(r_i-f_i-1)+(f_i-\ell_i) + 1=r_i-\ell_i$.
  \begin{figure}[htbp]
    \centering%
    \subfloat[before]{\includegraphics[scale=\figscale]{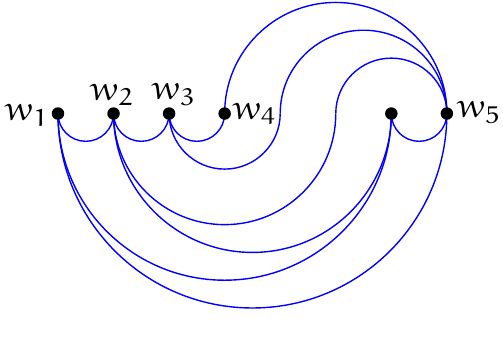}}\hfil
    \subfloat[after]{\includegraphics[scale=\figscale]{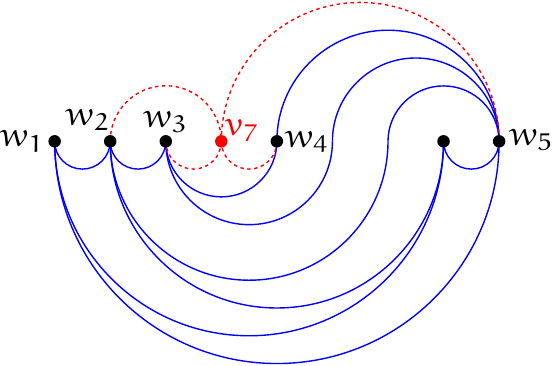}}\hfil
    \caption{Inserting a new vertex $v_7$ with $\ell_6=2$, $r_6=5$, and
      $f_6=3$.\label{fig:wfiexists}}
  \end{figure}

  In the second case there is no $w_{f_i}$ and $v_{i+1}$ is inserted
  just to the right of $w_{\ell_i}$. But we also know that none of the
  edges in the path $w_{\ell_i+1},\ldots,w_{r_i}$ of $C_i$ are below
  the spine. Therefore, all these edges will be proper arcs in the final drawing
  that have not been counted yet. Together with the new edge
  $v_{i+1}w_{\ell_i}$, which is drawn below the spine, we get
  again $r_i-\ell_i$ new proper arcs in the final drawing.

  In summary, we \emph{always} get at least $r_i-\ell_i$ new proper arcs in the
  final drawing when inserting a vertex $v_{i+1}$. In the last step, when
  inserting $v_n$, we even get $r_{n-1}-\ell_{n-1}+1$ new proper
  arcs. Therefore, the total number of proper arcs is bounded from below by
  $4+\sum_{i=3}^{n-1}(r_i-\ell_i)$. The total number of edges in $G$ is
  \[
  3n-6=3+\sum_{i=3}^{n-1}(r_i-\ell_i+1)\,.
  \]
  Combining both expressions yields at least
  \[
  4+\sum_{i=3}^{n-1}(r_i-\ell_i)=4+2n-6=2n-2
  \]
  proper arcs and, therefore, at most $3n-6-(2n-2)=n-4$ biarcs in the final
  drawing.

  Regarding the runtime bound, observe that when inserting a new vertex we
  inspect all its neighbors on the current outer cycle to select the right spot
  for insertion. Therefore the time spent for each vertex is proportional to its
  degree. As the graph is planar, the sum of all vertex degrees is linear. The
  arc diagram under construction can be represented as a tree using standard
  techniques~\cite{cp-ltadp-95}, where in addition we also store for every edge
  whether it is a proper arc or a down-up biarc.
\end{proof}

\section{General Biarc Diagrams}\label{sec:generalbiarcs}

In this section we discuss the connection between biarc diagrams and edge flips
and subdivisions in triangulations. Recall that Bernhart and
Kainen~\cite{bk-btg-79} characterized planar graphs that admit a plane proper
arc diagram as all subhamiltonian planar graphs.  The following theorem
generalizes this characterization in the context of biarc diagrams (the original
Theorem is obtained by setting $k=0$).

\bernkaingeneral*
\begin{proof}
  First, suppose there is a biarc diagram of $G$ with at most $k$ biarcs. Then
  we can simply subdivide these at most $k$ biarcs in order to obtain a proper
  arc diagram of some graph $G'$. By the characterization of Bernhart and
  Kainen, $G'$ is subhamiltonian.

  Second, fix a set of at most $k$ edges in $G$ so that subdividing them results
  in a subhamiltonian graph $G'$. By the characterization of Bernhart and Kainen
  we know that $G'$ admits a proper arc diagram. Removing the new vertices from
  the subdivided edges in that arc diagram results in a biarc diagram of $G$
  with at most $k$ biarcs (if both arcs incident to a subdivision vertex are on
  the same side of the spine, then the biarc can be replaced by a single proper
  arc).
\end{proof}

A similar statement can be obtained for simultaneous edge flips, where the edges
to be manipulated must not share a triangle. As this is a more restricted
setting, we get a correspondence in one direction only. But this is enough for
the purpose of getting upper bounds on the number of biarcs.

\begin{lemma}\label{lem:simflipbiarc}
  If a maximal planar graph $G$ can be transformed into a Hamiltonian graph
  with a simultaneous flip of $k$ edges, then $G$ admits a plane biarc diagram
  with at most $k$ biarcs.
\end{lemma}
\begin{proof}
  Let $H$ be a Hamiltonian graph obtained from $G$ by simultaneously
  flipping an edge set $E_1$ to $E_2$ with $|E_1|=k$. Without loss of
  generality, assume that $E_1$ is a minimal set of edges that must be
  flipped in order to obtain a Hamiltonian graph.  Consequently, every
  Hamiltonian cycle in $H$ passes through all $k$ edges in $E_2$.
  If we subdivide each edge in $E_2$, we obtain another Hamiltonian graph
  $H'$. Now consider the graph $G'$ obtained from $G$ by subdividing
  each edge in $E_1$, and identify the subdivision vertices of the
  corresponding edges in $G'$ and $H'$. Notice that the union of $G'$
  and $H'$ is a plane graph that contains $H'$, hence it is
  Hamiltonian. Consequently $G'$ is subhamiltonian. By
  Theorem~\ref{thm:bernkaingeneral}, $G$ admits a plane biarc diagram
  with at most $k$ biarcs, as claimed.
\end{proof}

In order to obtain a general statement about arc diagrams from
Lemma~\ref{lem:simflipbiarc}, we need a bound on the number of edges to
simultaneously flip in a given graph in order to make it Hamiltonian.  Even the
existence of such a simultaneous flip---regardless of the number of edges
involved---is not obvious to begin with. For instance, consider triangulations
$G_1$ and $G_2$ where $G_1$ has a vertex with linear degree and all vertices in
$G_2$ have constant degree (e.g., a nested triangle graph). As a single
simultaneous flip can only change about half of the edges incident to a vertex,
at least a logarithmic number of simultaneous flips is required to transform
$G_1$ into $G_2$~\cite{bcgmw-sdfpt-07}.

Bose et al.~\cite{bcgmw-sdfpt-07} showed that every triangulation on
$n\geq 6$ vertices can be transformed to a $4$-connected (hence
Hamiltonian) triangulation by a single simultaneous flip. However, no
bound is known on the number of flipped edges, which leaves us with
the trivial bound of $(2n-4)/2=n-2$. Note that the corresponding bound on
the number of biarcs is similar to the one from
Theorem~\ref{thm:biarcmonotone}, but there we could guarantee that all
biarcs are monotone. Using Lemma~\ref{lem:simflipbiarc} we do not have
any control over the type of biarcs used.

\section{Simultaneous Flip Distance to 4-connectivity}\label{sec:simflip}

In this section we determine the maximum number of edges needed to transform an
$n$-vertex triangulation into a $4$-connected triangulation using a single
simultaneous flip. Consider a triangulation $G=(V,E)$. As there is no
$4$-connected triangulation on fewer than six vertices, suppose that $G$ has at
least six vertices. We would like to transform $G$ into a $4$-connected
triangulation by simultaneously flipping a set $F\subset E$ of edges such that
all separating triangles are destroyed and none created. We use the following
criterion to ensure that the resulting triangulation is $4$-connected.
\begin{lemma}[Bose et al.~\cite{bcgmw-sdfpt-07}]\label{lem:simflip}
  Let $F$ be a set of edges in a triangulation $G$ such that no two edges in $F$
  are incident to a common triangle, every edge in $F$ is incident to a
  separating triangle, and for every separating triangle $T$ there is at least
  one edge in $F$ that is incident to $T$. Then $F$ is simultaneously flippable
  in $G$ and the resulting triangulation is $4$-connected.
\end{lemma}

Recall that the edges of a triangulation $G$ and its dual $G^*$ are in
one-to-one correspondence. Consequently, the set $F^*$ of edges dual to those in
$F$ forms a matching in $G^*$. As all faces of a triangulation are triangles,
$G^*$ is cubic ($3$-regular). Moreover, every triangulation on $n\ge 4$ vertices
is $3$-connected and so its dual is bridgeless ($2$-edge-connected). By a famous
theorem of Tait the following statement is equivalent to the Four-Color Theorem:
\begin{theorem}[Tait~\cite{bm-gt-2008}]\label{thm:tait4c}
  Every bridgeless cubic planar graph admits a partition of the edge set into
  three perfect matchings.
\end{theorem}
In particular, this applies to the dual of a triangulation. Call a set
$F\subseteq E$ of edges of a triangulation $G=(V,E)$ a (perfect) \emph{dual
  matching} if $F^*$ forms a (perfect) matching of $G^*$. While it is clear that
a perfect dual matching contains exactly one edge of each facial triangle, this
is not obvious for separating triangles. But it follows from a simple parity
argument, as the following lemma shows.\footnote{Bose et
  al.~\cite{bcgmw-sdfpt-07} derive this property from the explicit Tait
  coloring. The statement here is slightly more general because it holds for
  every perfect dual matching.}

\begin{restatable}{lemma}{ldualhittriangle}\label{p:dualhittriangle}
 Every perfect dual matching of a triangulation $G$ contains an edge of
 every triangle of $G$.
\end{restatable}
\begin{proof}
 For facial triangles the statement holds by definition. So consider a
 separating triangle $T$ of $G$ and the subgraphs $H$ and $H'$ of $G$ induced by $T$
 together with the two respective components of $G\setminus T$. As $H$ is a maximal
 planar graph, it has $2|\mathrm{V}(H)|-4$ faces including the facial triangle
 $T$. Hence the number of faces of $H$ different from $T$ is odd and so every perfect
 matching of $G^*$ contains at least one edge that connects a face of $H$ with
 a face of $H'$. The corresponding primal edge of the dual matching
 is an edge of $T$, as required.
\end{proof}

The combination of Theorem~\ref{thm:tait4c} with Lemma~\ref{p:dualhittriangle}
immediately yields the following
\begin{restatable}{corollary}{cdualhittriangle}\label{c:dualhittriangle}
  Every triangulation $G$ admits a partition of the edge set into three perfect
  dual matchings such that every triangle of $G$ is incident to exactly one edge
  from each of the three matchings.
\end{restatable}

The last missing bit to prove Theorem~\ref{thm:upsim4} is an upper bound on the
number of edges in a triangulation that can be incident to separating triangles.
\begin{restatable}{lemma}{edgebound}\label{p:edgebound}
  At most $2n-7$ edges of a maximal planar graph on $n\ge 4$ vertices are
  incident to separating triangles. This bound is the best possible.
\end{restatable}
\begin{proof}
  We proceed by induction on the number of separating triangles. For a maximal planar
  graph without separating triangles the statement is trivial. For $n=4$,
  the only maximal planar graph is $K_4$ and it has no separating
  triangle. For $n=5$, there is only one maximal planar graph up to isomorphism,
  and it contains exactly one separating triangle, bounded by $3=2\cdot 5-7$ edges.

  Consider a maximal planar graph $G$ on $n\ge 6$ vertices and a minimal
  separating triangle $T$ of $G$, that is, a separating triangle such that for
  at least one component $C$ of $G\setminus T$ the subgraph $H:=G[C\cup T]$ does
  not contain a separating triangle (equivalently, $H=K_4$ or $H$ is
  $4$-connected). Put $k=|C|\in\{1,\ldots,n-4\}$. The graph $G'=G\setminus C$
  has $n-k$ vertices and contains exactly one fewer separating triangle than
  $G$. By the inductive hypothesis, at most $2(n-k)-7$ edges of $G'$ are
  incident to separating triangles of $G'$. As far as the corresponding count
  for $G$ is concerned, only the three edges of $T$ have to be accounted for in
  addition.

  If some edge of $T$ also bounds a separating triangle
  in $G'$, then this edge has already been counted inductively in $G'$.
  Including the remaining at most two edges of $T$, we see that
  at most $2(n-k)-7+2\le 2n-7$ edges of $G$ are incident to separating
  triangles of $G$. Also if $k\ge 2$, then at most $2(n-k)-7+3\le 2n-8$ edges of
  $G$ are incident to separating triangles of $G$.

  Otherwise, $k=1$ and none of the edges of $T$ is incident to any separating
  triangle in $G'$. Denote the vertices of $T$ by $T=(a,b,c)$, and let $T$ and
  $(b,a,d)$ be the two faces of $G'$ incident to the edge $ab$. By contracting
  the edge $ab$ in $G'$, we obtain a graph $G''$ on $n-2$ vertices. The
  contraction identifies the two edges $ac$ and $bc$ into a single edge.
  Similarly the two edges $ad$ and $bd$ are identified into a single edge.

  We claim that after this contraction $G''$ is simple, that is, no multi-edge
  is introduced (other than the two edge pairs already mentioned and
  handled). This is because the vertices $a$ and $b$ have exactly two common
  neighbors in $G'$, which are $c$ and $d$. If $a$ and $b$ had any other common
  neighbor $w\notin\{c,d\}$, then the triangle $abw$ would be a separating
  triangle in $G'$, contrary to our assumption that $ab$ is not incident to any
  separating triangle in $G'$. Hence $b$ and $d$ are the only common neighbors
  of $a$ and $b$, and so no multi-edge is created by contracting $ab$, as
  claimed.

  Finally we observe that by the inductive hypothesis at most $2(n-2)-7=2n-11$
  edges of $G''$ are incident to separating triangles of $G''$. In addition to
  the three edges of $T$ we also have to account for changes caused by the
  contraction of the edge $ab$. Edges $ac$ and $bc$ are identified in $G''$, but
  neither is incident to any separating triangle in $G'$ by assumption. Edges
  $ad$ and $bd$ are also identified. They each may be incident to separating
  triangles in $G'$ but they are counted once only in $G''$. Consequently, we
  count the three edges of $T$ and one additional edge for a total of at
  most $(2n-11)+3+1=2n-7$ edges incident to separating triangles in $G$.

  For a matching lower bound, consider the graphs depicted in
  \figurename~\ref{fig:edgebound}.  The solid edges are incident to separating
  triangles. On the left, we have $n=6$ and $2n-7=5$ and exactly $5$ edges
  incident to separating triangles. To obtain larger examples, repeatedly insert
  a new vertex into a face with exactly one solid edge. The remaining two edges
  of this face become solid. Note that this operation creates a face with
  exactly one edge that is incident to a separating triangle, and so the
  operation can be repeated indefinitely. After $k$ such operations we have
  $n=6+k$ vertices and precisely $5+2k = 2n-7$ edges incident to separating
  triangles, as desired.
  \begin{figure}[htbp]
    \centering
    \includegraphics{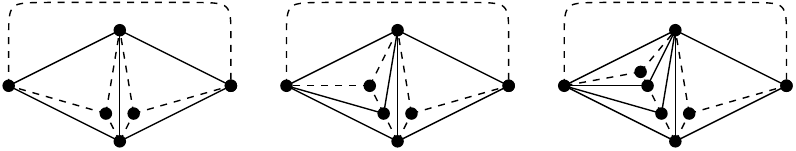}
    \caption{Tight examples for Lemma~\ref{p:edgebound}, for $n=6,7,8$.}
    \label{fig:edgebound}
  \end{figure}
\end{proof}
Now we have all pieces together to prove Theorem~\ref{thm:upsim4}.
\upsimfour*
\begin{proof}
  Consider a maximal planar graph $G$ on $n$ vertices.  By
  Corollary~\ref{c:dualhittriangle} the $3n-6$ edges of $G$ can be partitioned
  into three perfect dual matchings $D_1$, $D_2$, and $D_3$, of $n-2$ edges
  each, such that each separating triangle is incident to one edge from each.
  Let $M_i$, for $i\in\{1,2,3\}$, denote the dual matching that results from
  removing all edges from $D_i$ that are not incident to any separating
  triangle. By Lemma~\ref{p:edgebound} at most $2n-7$ edges of $G$ are incident
  to separating triangles. Therefore, one of $M_1$, $M_2$, and $M_3$ contains at
  most $\lfloor(2n-7)/3\rfloor$ edges. By Lemma~\ref{lem:simflip} these edges
  are simultaneously flippable and the resulting graph is $4$-connected.

  All separating triangles (and incident edges) can be found in $O(n)$
  time~\cite{cn-asla-85}. Theorem~\ref{thm:tait4c} is known to be equivalent to
  the Four Color Theorem~\cite{bm-gt-2008}, and a proper $4$-coloring of $G$
  yields an edge partition into dual matchings in all 4-connected subgraphs in
  $O(n)$ time.  The current best algorithm for $4$-coloring a planar graph with
  $n$ vertices runs in $O(n^2)$ time~\cite{rsst-fct-97}. Consequently, we can
  find a smallest dual matching from $\{M_1,M_2,M_3\}$ in $O(n^2)$ time.
\end{proof}

The following construction shows that the bound in Theorem~\ref{thm:upsim4} is
tight up to an additive constant of $\pm 1$. %
\upsimfourlower*
\begin{proof}
  Start with $G_0=K_4$ and select a face $f_0$ of $G_0$. For $i\in\N$, the graph
  $G_i$ is recursively obtained from $G_{i-1}$ as follows (see \figurename~\ref{fig:3}
  where $f_0$ is the outer face): For each face $f$ adjacent to $f_0$ in $G_{i-1}$, insert
  a new vertex of degree 3 into $f$ and connect it to all three vertices of $f$.
  Since $f_0$ is adjacent to three distinct faces, the number of vertices in $G_i$ is
  $n_i=3i+4$. By construction, $G_i$ has three groups of separating triangles.
  Each group contains $i$ separating triangles that lie in one of the three
  subdivided faces of $G_0$ and share a common edge with $f_0$.

  As the face $f_0$ is incident to all $3i$ separating triangles in $G_i$, it is
  tempting to just flip the three edges of $f_0$. However, a simultaneous flip
  can include at most one of the edges incident to $f_0$. Consequently, at least
  two edges of $f_0$ remain untouched, each of which is incident to a group of
  $i$ separating triangles. As no two triangles within a group share any other edge,
  one flip per triangle is needed to destroy them all simultaneously. Also two
  separating triangles from different groups are edge-disjoint---except for the three
  largest separating triangles, which are bounded by the edges of $G_0$.
  But any two groups share only one such edge and so at most one flip can be saved
  in this way. Therefore, in order to handle the two groups whose edge incident to
  $f_0$ is not flipped at least $2i-1$ edges need to be flipped. Clearly, at least
  one more edge flip is required to handle the third group, which leaves us with
  the claimed bound of at least $2i$ edges.
\end{proof}
\begin{figure}[htbp]
  \centering%
  \subfloat[$G_1$]{\includegraphics[scale=\figscale]{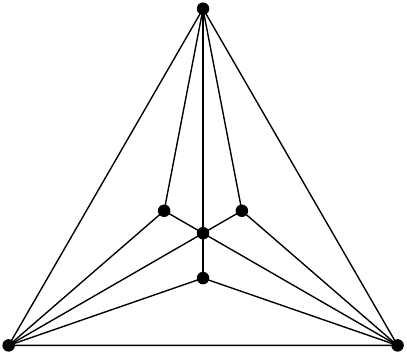}}\hfil
  \subfloat[$G_2$]{\includegraphics[scale=\figscale]{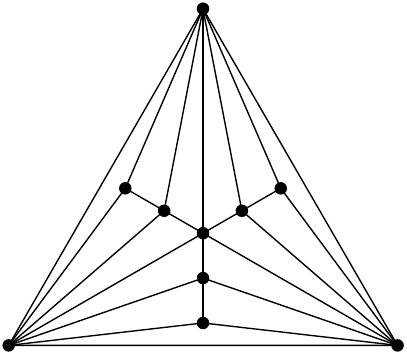}}\hfil
  \subfloat[$G_4$]{\includegraphics[scale=\figscale]{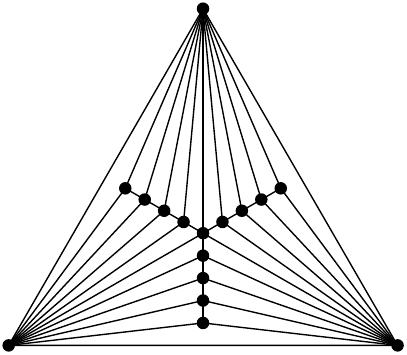}}\hfil
  \caption{The first members of a family of triangulations that require a
    simultaneous flip of at least $(2n-8)/3$ edges to become
    $4$-connected.\label{fig:3}}
\end{figure}

\section{Flip Distance to Hamiltonicity}\label{sec:hamflip}

With regard to arc diagrams, there is actually no reason to insist that the
triangulation be $4$-connected. In order to apply Lemma~\ref{lem:simflipbiarc}
we need only that the triangulation is Hamiltonian. In this section we go one
step further and in addition lift the restriction that the flip be
simultaneous. Instead, an arbitrary sequence of edge flips is allowed. In this
case tight bounds are known if the goal is to obtain a $4$-connected
triangulation. Bose at al.~\cite{bjrsv-mt4uf-14} showed that
$\lfloor(3n-9)/5\rfloor$ flips are always sufficient and sometimes $(3n-10)/5$
flips are necessary to transform a given triangulation on $n$ vertices into a
$4$-connected triangulation.

In general, a sequence  of flips has no direct implication for arc
diagrams. But if only edges of the original triangulation are flipped, then we
can subdivide those edges rather than flipping them. In the resulting arc
diagram only the subdivided edges may appear as biarcs. But a bound on the flip
distance to a Hamiltonian triangulation is of independent interest. For
instance, it is directly related to the current best upper bound on the
diameter of the flip graph of combinatorial
triangulations~\cite{bjrsv-mt4uf-14,k-dfts-96,mno-dfhts-03}. The argument uses a
single so-called canonical triangulation and shows that every triangulation can
be transformed into this canonical triangulation in two steps: First at most
$\lfloor(3n-9)/5\rfloor$ flips are needed to obtain a $4$-connected
triangulation and then an additional at most $2n-15$ flips are needed to
transform any $4$-connected triangulation into the canonical one. Combining two
such flip sequences yields an upper bound of $5.2n-33.6$ on the diameter of the
flip graph~\cite{bjrsv-mt4uf-14}. The bound of $2n-15$ flips for the second step
is actually tight~\cite{k-dfts-96}. The corresponding bound for a triangulation
that is Hamiltonian (but not necessarily $4$-connected) is slightly worse only:
It can be transformed into the canonical triangulation using at most $2n-10$
flips~\cite{mno-dfhts-03}. Hence our focus is to improve the first step
by showing that fewer flips are needed to guarantee a Hamiltonian triangulation
than a $4$-connected one.

\hamflip*

\minipar{Proof outline.} The proof is constructive and consists of two
steps. In a first step we apply a sequence of elementary operations that
transform a triangulation $G$ into a $4$-connected triangulation $G'$. An
elementary operation is either a usual edge flip or a \emph{dummy flip}, where
a facial triangle $T$ is subdivided into three triangles by inserting a new
(dummy) vertex and then all three edges of $T$ are flipped. All this will be
done in such a way that $G'$ becomes $4$-connected and, therefore, contains a
Hamiltonian cycle $H'$. We then remove all dummy vertices and construct a
Hamiltonian cycle $H''$ resembling $H'$ in the resulting triangulation $G''$.
Finally, we argue that $G''$ can be obtained from $G$ with at most $n/2$ (usual)
edge flips. Specifically, we show that each dummy flip can be implemented using
at most two edge flips.

\minipar{Dummy flips.} Given a triangulation $G$ on $n\ge 4$ vertices and a
facial triangle $T$ of $G$, a dummy flip of $T$ transforms $G$ as follows
(\figurename~\ref{fig:5}): First, insert a new (dummy) vertex $v$ in the interior of face
$T$ and connect it to all three vertices of $T$. Note that $T$ becomes a separating
triangle in the resulting graph. Second, flip all three edges of $T$ in an arbitrary order.
\begin{figure}[htbp]
  \centering
  \subfloat[before]{\includegraphics[scale=\figscale]{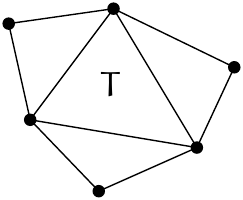}}\hfil
  \subfloat[Step~1]{\includegraphics[scale=\figscale]{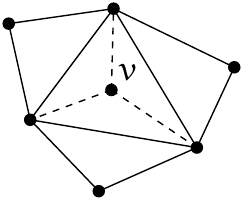}}\hfil
  \subfloat[Step~2]{\includegraphics[scale=\figscale]{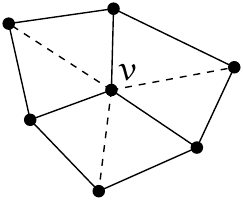}}\hfil
  \caption{Example of a dummy flip.\label{fig:5}}
\end{figure}
Similarly to the usual flip operation, a dummy flip may create multiple edges.
But we will use this operation in specific situations only---as specified in the
lemma below---where we can show that it produces a triangulation (that is, no
multiple edges).
\begin{restatable}{lemma}{dummyflip}\label{prop:dummy-flip}
  Let $G$ be a maximal planar graph and let $T$ be a facial triangle of $G$ such
  that every edge of $T$ is incident to a separating triangle of $G$. Then the
  dummy flip operation of $T$ in $G$ produces no parallel edges and no new
  separating triangles.
\end{restatable}
\begin{proof}
  Let $T=abc$ be a facial triangle of $G$ as specified above and insert a new
  vertex $v$ into $T$. First we claim that every edge of $T$ is flippable.
  Consider the edge $ab$ and assume that it is incident to faces $T'=abv$ and
  $T''=abd$. The only obstruction to flippability of $ab$ is the presence of an
  edge $vd$ in $G$. By assumption there is a separating triangle $S=abe$ in $G$,
  for some vertex $e$. Given that both $T'$ and $T''$ are facial, the vertices
  $v$ and $d$ are separated by $S$ (they are in different components of
  $G\setminus S$). Therefore, by planarity of $G$, the edge $vd$ is not present
  in $G$ and so $ab$ is flippable, as claimed.

  Noting that any two distinct triangles in a (simple) graph share at most one
  edge, we observe that no separating triangle shares two edges with $T$. In
  particular, flipping the edge $ab$ does not destroy any separating triangle
  incident to the edges $bc$ or $ca$. Hence even after flipping one or two edges
  of $T$, we can still apply the above reasoning to show that the other edge(s)
  of $T$ remain flippable. It follows that all three edges of $T$ can be flipped
  in any order.

  It remains to show that these flips do not introduce any separating
  triangle. Denote by $G'$ the graph that results from the dummy flip of $T$ in
  $G$. As all newly introduced edges are incident to $v$, any new separating
  triangle must also be incident to $v$. So suppose $S=vwx$ is a separating
  triangle in $G'$. In particular, this means that the edge $wx$ was present in
  $G$ already. At most one of $w$ and $x$ can be vertices of $T$, otherwise
  $wx$ would not be an edge of $G'$ (exactly the edges of $T$ were flipped away,
  after all). So we may suppose without loss of generality that $w$ is a vertex
  of some triangle $T'=abw$ in $G$. However, by assumption there is a separating
  triangle incident to the edge $ab$ in $G$, which separates $w$ (in $G$) from
  all neighbors of $v$ in $G'$ other than $a$ and $b$. It follows that
  $x\in\{a,b\}$, but the triangles $vwa$ and $vwb$ are facial in $G'$ by
  construction. Therefore, there is no separating triangle in $G'$ that is
  incident to $v$ and so no separating triangle is introduced by the dummy flip
  of $T$ in $G$.
\end{proof}

\subsection{First Step: Establish 4-Connectedness}

Our main lemma to establish Theorem~\ref{thm:hamflip} is the following.
\begin{lemma}\label{lem:charging}
  Every maximal planar graph on $n\ge 6$ vertices can be transformed into a
  $4$-connected maximal planar graph by a sequence of $f$ flip and $d$
  dummy flip operations, for some $f,d\in\N$, such that $f+2d\le (n-3)/2$.
\end{lemma}
Recall that there are triangulations on $n$ vertices that contain
$\lfloor(3n-9)/5\rfloor$ pairwise edge-disjoint separating
triangles~\cite{bjrsv-mt4uf-14,hhs-s3pt-01}. In this case, we need to flip away
at least one edge from each separating triangle to reach $4$-connectivity.
Considering that a dummy flip operation flips three edges, we must have
$f+3d\ge\lfloor(3n-9)/5\rfloor$. The crucial claim in Lemma~\ref{lem:charging}
is that $f+2d\le(n-3)/2$ is possible, and later we will show how to replace each
dummy flip by two usual flips rather than three (Lemma~\ref{lem:eliminate}).

The rest of this section is devoted to the proof of Lemma~\ref{lem:charging}. We
describe an algorithm that, given a triangulation $G$ on $n\geq 6$ vertices,
returns a sequence of $f$ flip and $d$ dummy flip operations that produces a
$4$-connected graph. The bound is written equivalently as $6f+12d\le
3n-9$ and is established via the following
charging scheme. Each edge of $G$, with the exception of the three edges of the
outer face, receives one unit of credit. Each edge flip costs six units.
Each dummy flip costs fifteen units and produces three new edges, each of
which receives one unit of credit.

\minipar{4-Block Decomposition.} In our algorithm, we recursively process
$4$-connected subgraphs using the $4$-block tree $\mathcal{B}$ of $G$ (see
\figurename~\ref{fig:4block} for an example). By fixing an (arbitrary) plane
embedding of $G$, we make $\mathcal{B}$ a rooted tree such that the root is the
$4$-block that contains the boundary of the outer face of $G$. Every separating
triangle $T$ of $G$ corresponds to an edge between two $4$-blocks, where the
parent lies in the exterior of $T$ (plus $T$) and the child lies in the interior
of $T$ (plus $T$). For a $4$-block $G_i$ in $\mathcal{B}$ denote by $T_i$ the
outer face of $G_i$, and denote by $n_i$ the number of vertices of $G_i$ minus
three (the vertices of $T_i$). As a maximal planar graph, $G_i$ has
$3(n_i+3)-6=3(n_i+1)$ edges and $2(n_i+3)-4=2(n_i+1)$ faces. An edge of $G_i$ is
called an \emph{interior} edge if it is not incident to the outer face
$T_i$. For each $4$-block $G_i$ in $\mathcal{B}$ we maintain counters $f_i$ and
$d_i$ that denote the number of flips and dummy flips, respectively, that were
used within $G_i$ during the course of the algorithm. Initially $f_i=d_i=0$, for
every vertex $G_i$ of $\mathcal{B}$.

\begin{figure}[htbp]
  \centering%
  \subfloat[before flipping $e$]{\includegraphics{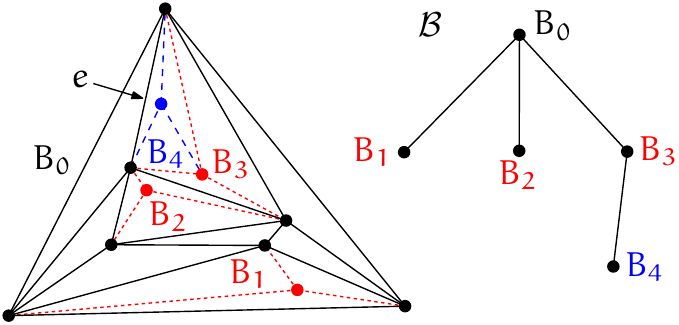}}\hfil
  \subfloat[after flipping $e$]{\includegraphics{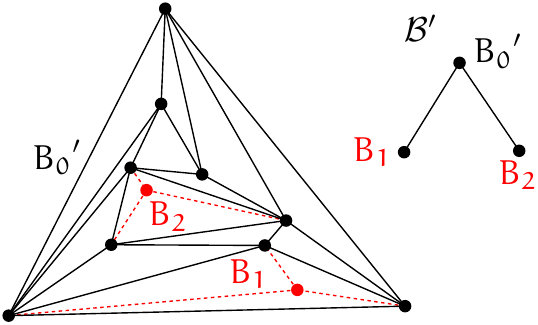}}
  \caption{Example of a 4-block decomposition and how a flip of the edge $e$
    merges blocks. The vertices and edges of the root (level zero) are shown
    solid black, the vertices and edges on level one are shown dotted red, and
    the vertices and edges on level two are shown dashed
    blue.\label{fig:4block}}
\end{figure}

The algorithm computes the sequence of flip and dummy flip operations
incrementally, and maintains a current triangulation produced by the
operations. Both the graph $G$ and the $4$-block decomposition $\mathcal{B}$
change dynamically during the algorithm: when we flip an edge $e$ of some
separating triangle(s), all $4$-blocks containing edge $e$ merge into a single
$4$-block.  At the end of the algorithm, the tree $\mathcal{B}$ consists of a
single $4$-block that corresponds to the $4$-connected graph $G'$.
In order to avoid notational clutter, we always denote the current $4$-block
tree by $\mathcal{B}$. As an invariant (detailed below) we maintain that at each
node of $\mathcal{B}$ the number of interior edges (ignoring dummy edges)
balances the cost of operations that were spent in this $4$-block. As
$\mathcal{B}$ evolves, so does the graph $\mathcal{G(B)}$ \emph{represented} by
$\mathcal{B}$. This graph is the union of all nodes ($4$-blocks) in
$\mathcal{B}$, where for any edge of $\mathcal{B}$ the vertices and edges of the
common triangle in the two endpoints ($4$-blocks) are identified.

\minipar{Main loop.} At every step, we take an arbitrary $4$-block $G_i$ on the
penultimate level of $\mathcal{B}$, that is, $G_i$ is not a leaf but all of its
children are leaves. Let $C_i$ denote the set of indices $c$ such that $G_c$ is
a child of $G_i$ in $\mathcal{B}$, and denote $\mathcal{T}_i=\{T_c\mid c\in
C_i\}$. The algorithm selects a sequence of edges of $G_i$ to be flipped (or
dummy flipped) in order to merge $G_i$ with $G_c$, for all $c\in C_i$, into a
new $4$-block $G_z$. Denote the resulting $4$-block tree by $\mathcal{B}'$. If
no edge of $T_i$ is flipped, then $G_z$ is a leaf of $\mathcal{B}'$. But if an
edge of $T_i$ is flipped, then $G_z$ may be an interior node of $\mathcal{B}'$.

\minipar{Algorithmic preliminaries.}  In each iteration, we flip the edges of a
dual matching of $G_i$ (a $4$-connector, defined below), but if $\mathcal{T}_i$
forms a checkerboard (defined below), we substitute three of these flip
operations by one dummy flip.

A \emph{$4$-connector for $G_i$} is a dual matching of $G_i$ that contains
precisely one edge from every triangle in $\mathcal{T}_i$. By
Lemma~\ref{lem:simflip} we can flip the edges of a $4$-connector in an arbitrary
order, and the 4-blocks $G_c$, for all $c\in C_i$, merge into $G_i$. Note that a
perfect dual matching for $G_i$ consists of $2(n_i+1)/2=n_i+1$ edges and so
every $4$-connector contains at most this many edges.

Consider a partition of the edge set of $G_i$ into three perfect dual matchings
$D_1$, $D_2$, and $D_3$ (Theorem~\ref{thm:tait4c}). For each $D_i$,
$i\in\{1,2,3\}$, the subset $M_i$ of edges that are incident to some triangle
from $\mathcal{T}_i$ is a $4$-connector for $G_i$.  We select
$M\in\{M_1,M_2,M_3\}$ according to the following criteria:
\begin{itemize}
\item $M$ has minimum cardinality and
\item if possible (among the sets of minimum cardinality), then $M$ contains an
  edge of $T_i$.
\end{itemize}
Every $4$-connector that is obtained from some partition $D_1$, $D_2$, $D_3$ in
the described way is an \emph{optimal $4$-connector} for $G_i$ in
$\mathcal{B}$.

We say that $\mathcal{T}_i$ is a \emph{checkerboard} if every interior edge of
$G_i$ belongs to exactly one triangle of $\mathcal{T}_i$. If $\mathcal{T}_i$ is
a checkerboard, then we perform a dummy flip on a triangle $F$ that is selected
according to the following lemma (see \figurename~\ref{fig:lemma17} for
illustration).

\begin{restatable}{lemma}{prelimdummy}\label{lem:predummy}
  If $\mathcal{T}_i$ is a checkerboard, then $G_i$ contains two triangles, $F$
  and $H$, such that $F$ is a bounded facial triangle adjacent to three
  triangles in $\mathcal{T}_i$ and $H$ is adjacent to $T_i$ but not to $F$.
\end{restatable}
\begin{proof}
  We partition the set of facial triangles of $G_i$ into two subsets: the set
  $\mathcal{T}_i$ (which are separating triangles in
  $\mathcal{G}(\mathcal{B})$), and the set of all other faces that we denote by
  $\mathcal{F}_i$.  The dual graph $G_i^*$ is a $3$-regular planar graph on
  $2n_i+2$ nodes, one of which corresponds to the outer face $T_i$.

  If $\mathcal{T}_i$ is a checkerboard, then the $2n_i+1$ bounded faces of $G_i$
  induce a bipartite subgraph in $G_i^*$ between $\mathcal{T}_i$ and the bounded
  faces in $\mathcal{F}_i$. This subgraph has precisely three vertices of degree
  two (adjacent to the outer face), all other degrees are three. Since the sum
  of degrees in the two vertex classes are equal, all three neighbors of the
  outer face must be in the same vertex class. Therefore, $G_i^*$ is a bipartite
  graph on all faces, where the two classes are either $\mathcal{T}_i$ and
  $\mathcal{F}_i$, or $\mathcal{T}_i\cup\{T_i\}$ and
  $\mathcal{F}_i\setminus \{T_i\}$. Given that $G_i$ has $2n_i+2$ faces
  (including the outer face $T_i$), the two classes each have size $n_i+1$.

  In particular, $T_i$ is adjacent to three distinct facial triangles of $G_i$
  that are either all in $\mathcal{T}_i$ or all in $\mathcal{F}_i$. We
  distinguish two cases.  First assume $T_i$ is adjacent to three triangles in
  $\mathcal{T}_i$.  Then $\mathcal{F}_i$ also contains at least three
  triangles. Since $G_i^*$ is planar, it does not contain $K_{3,3}$ as a
  subgraph, and so there exists a bounded face $F\in \mathcal{F}_i$ that is not
  adjacent to all three triangles adjacent to $T_i$, and a face
  $H\in \mathcal{T}_i$ adjacent to $T_i$ but not to $F$.  Next assume $T_i$ is
  adjacent to three triangles in $\mathcal{F}_i$, let one of them be $H$.  Since
  these triangles are edge-disjoint, each vertex of $T_i$ is incident to a
  distinct triangle in $\mathcal{T}_i$. This implies that
  $|\mathcal{T}_i\cup \{T_i\}|\geq 4$, and so there is a fourth triangle
  $F\in \mathcal{F}_i\setminus \{T_i\}$ that is adjacent to three triangles in
  $\mathcal{T}_i$.
\end{proof}

\begin{figure}[htbp]
  \centering%
  \subfloat[]{\includegraphics{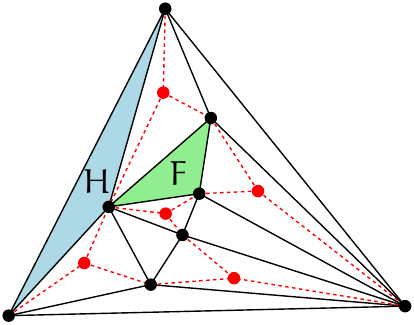}}\hfil
  \subfloat[]{\includegraphics{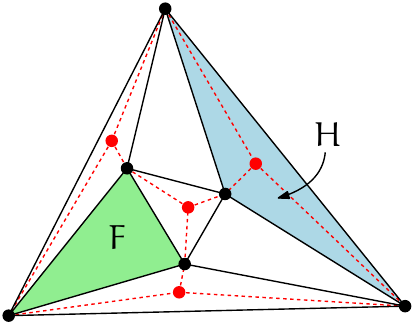}}\hfil
  \subfloat[]{\includegraphics{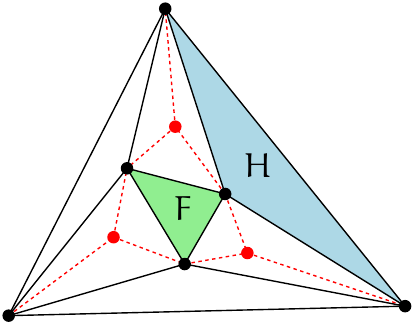}}
  \caption{Three examples to illustrate Lemma~\ref{lem:predummy}. The vertices and
    edges of $G_i$ are shown solid black, the vertices and edges of its children
    are shown dotted red.\label{fig:lemma17}}
\end{figure}

\minipar{Algorithm {\sc 4Connect}$(G)$.}  Given a triangulation $G$, fix an
arbitrary embedding of $G$.  This embedding defines a rooted $4$-block tree
$\mathcal{B}$. While $\mathcal{B}$ is not a singleton, do:
\begin{enumerate}[label={(\arabic*)}]\setlength{\itemindent}{2\labelsep}
\item Consider an arbitrary vertex $G_i$ at the penultimate level of
  $\mathcal{B}$.
\item If $\mathcal{T}_i$ is \emph{not} a checkerboard, then find an optimal
  $4$-connector $M$ for $G_i$ and flip the edges of $M$ in an arbitrary order.
\item Otherwise, let $F$ and $H$ be two triangles of $G_i$ as in
  Lemma~\ref{lem:predummy}.  Let $D\in \{D_1,D_2,D_3\}$ be the dual perfect
  matching that contains the common edge of $H$ and $T_i$.  First apply a dummy
  flip to $F$. Then consider all triangles in $\mathcal{T}_i$ that are not
  adjacent to $F$, in an arbitrary order. For every such triangle, flip the
  incident edge in $D$.
\item Finally, update $\mathcal{B}$ and $\mathcal{G}(\mathcal{B})$.
\end{enumerate}

\minipar{Correctness of the Algorithm.}  We show that the above algorithm
transforms an input triangulation $G$ on $n$ vertices into a $4$-connected
triangulation using a sequence of $f$ flips and $d$ dummy flips, for some
$f,d\in\N_0$, such that $f+2d\leq (n-3)/2$.
By Lemmata~\ref{lem:simflip}, \ref{prop:dummy-flip}, and \ref{lem:predummy}, the
operations described in the algorithm can be performed. In every step of the
algorithm at least two nodes of the $4$-block tree are merged. Therefore, after
a finite number of steps we are left with a block tree that consists of a single
$4$-block $G'$.

\minipar{Independent dummy vertices.} The following observation is crucial for
the second step of our algorithm (Section~\ref{sec:algstep2}) where we eliminate
dummy vertices and simulate dummy flips using regular edge flips.
\begin{restatable}{observation}{obsdummy}\label{obs:dummy}
  For each vertex $v$ created by a dummy flip operation in {\sc 4Connect}$(G)$,
  subsequent operations do not modify the six facial triangles incident to $v$.
\end{restatable}
\begin{proof}
  The claim directly follows from the following properties of the operations
  performed by the algorithm. (i) When the algorithm flips an edge (including
  the three flips of a dummy flip), this edge is incident to a separating
  triangle of the current graph. (ii) The algorithm never creates new separating
  triangles. (iii) For every vertex $v$ created by a dummy flip, at the end of
  this dummy flip none of the edges of the six triangles incident to $v$ is
  incident to any separating triangle in $\mathcal{G(B)}$.

  The first two properties are obvious, but the third may need a bit of
  justification: Each of the three edges of the face $F$ where $v$ is inserted
  is incident to a triangle from $\mathcal{T}_i$. In particular, the three
  neighbors of $v$ other than the vertices of $F$ lie inside these triangles and
  so do all edges between them and the vertices of $F$. By choice of $G_i$, the
  graph inside any triangle from $\mathcal{T}_i$ is a leaf of $\mathcal{B}$ and,
  therefore, does not contain any separating triangle of $\mathcal{G(B)}$.
\end{proof}

\minipar{Free and trapped edges.}  It remains to bound the number of flip and
dummy flip operations performed by the algorithm. An edge within some $4$-block
$G_i$ of $\mathcal{B}$ is \emph{free} if it is not incident to any separating
triangle of $\mathcal{G(B)}$. Free edges are a good measure of progress for our
algorithm because our final goal is to arrive at a state where all edges of
$\mathcal{G(B)}$ are free. An edge of $G_i$ that is not free is incident to one
or two triangles from $\mathcal{T}_i$. We refer to these edges as \emph{singly
  trapped} and \emph{doubly trapped}, respectively.

\minipar{Invariants.} As an invariant we maintain that every vertex $G_i$ of
$\mathcal{B}$ satisfies the following conditions:
\begin{enumerate}[label={(F\arabic*)}]\setlength{\itemindent}{3\labelsep}
\item\label{i:root} If $G_i$ is the only vertex of $\mathcal{B}$, then it has at
  least $6f_i+15d_i+3$ free edges.
\item\label{i:gen} If $G_i$ is a leaf of $\mathcal{B}$ that is not the root of
  $\mathcal{B}$, then $G_i$ has at least $6f_i+15d_i+3$ free interior edges.
\item\label{i:k4} If $G_i$ is an interior vertex of $\mathcal{B}$, then either
  $f_i=d_i=0$ or $G_i$ has at least $6f_i+15d_i+1$ free interior edges.
\end{enumerate}
Initially, \ref{i:root} holds since $\mathcal{B}$ has at least two vertices.
\ref{i:gen} holds for every leaf $G_i$ of $\mathcal{B}$ because all of the
interior $3(n_i+1)-3=3n_i$ edges are free, $n_i\ge 1$, and $f_i=d_i=0$. Finally,
\ref{i:k4} holds for every interior vertex $G_i$ of $\mathcal{B}$ because
$f_i=d_i=0$. Having a certain number of edges in a plane graph implies having a
certain number of vertices, as quantified by the following lemma.
\begin{restatable}{lemma}{pni}\label{p:ni}
  If $\mathcal{B}$ has at least two nodes, then $n_i\ge 2f_i+5d_i+1$,
  for every $4$-block $G_i$ in $\mathcal{B}$.
\end{restatable}
\begin{proof}
  For a leaf $G_i$ of $\mathcal{B}$, condition \ref{i:gen} implies that $G_i$
  has at least $6f_i+15d_i+6$ edges (the three edges of $T_i$ are not
  interior). As $G_i$ has exactly $3(n_i+1)$ edges, it follows that $n_i\ge
  2f_i+5d_i+1$. Similarly for an interior vertex $G_i$ of $\mathcal{B}$ with
  $f_i+d_i>0$, condition \ref{i:k4} implies that $G_i$ has at least
  $6f_i+15d_i+4$ edges and so $n_i\ge 2f_i+5d_i+1/3$. As $n_i$ is integral, we
  again obtain $n_i\ge 2f_i+5d_i+1$. Finally, if $f_i=d_i=0$, then the statement
  becomes $n_i\ge 1$, which is trivial.
\end{proof}

\minipar{Invariant maintenance.} It remains to show that each step of the
algorithm maintains invariants \ref{i:root}--\ref{i:k4}.  If an edge $e$ of
$T_i$ is flipped and $G_i$ is not the root of $\mathcal{B}$, then more blocks
may merge into $G_z$: The edge $e$ is definitely shared with the parent of $G_i$
in $\mathcal{B}$, but it may be shared with further ancestors as well. In
addition, the edge $e$ may belong to (at most) one sibling $G_s$ of $G_i$ and
possibly some descendants of $G_s$. We denote by $J$ the set of all $j$ such
that $G_j$ is a leaf of $\mathcal{B}$ that is merged into $G_z$. Similarly,
denote by $Q$ the set of all $q$ such that $G_q$ is an interior vertex of
$\mathcal{B}$ that is merged into $G_z$, and denote by $Q^+$ the set of indices
$q\in Q$ such that $f_q+d_q>0$. Note that neither $J$ nor $Q$ are empty, because
$C_i\subseteq J$ and $i\in Q$. However, we may have $Q^+=\emptyset$.

At the end of a step that merged all $G_j$, for $j\in J\cup Q$, into $G_z$ we
have $f_z=f+\sum_{j\in J\cup Q}f_j$ and $d_z=d+\sum_{j\in J\cup Q}d_j$, where
$f$ and $d$ denote the number of flips and dummy flips, respectively, that were
executed during this step. The following two lemmata do not make specific
assumptions about the set of operations (other than that they are valid, that
is, yield a triangulation).  In particular, the set of edges flipped need not
form a \emph{optimal} $4$-connector.
\begin{restatable}{lemma}{propcounttwo}\label{prop:count2}
  Suppose that $G_i$ together with all its children in $\mathcal{B}$ is merged
  into a \textbf{leaf} $G_z$ of $\mathcal{B}'$ using $f$ flips and $d$ dummy
  flips.  Then $G_z$ contains at least
  $6(f_z-f_i-f)+15(d_z-d_i-d)+3n_i+3|C_i|+3|Q|-3$ free interior edges.
\end{restatable}
\begin{proof}
  Combining Lemma~\ref{p:ni} with $C_i\subseteq J$ we obtain
  \begin{align}
    n_z = \sum_{j\in J\cup Q}n_j &\ge n_i+\sum_{j\in(J\cup
      Q)\setminus\{i\}}(2f_j+5d_j+1)\nonumber\\
    &\ge n_i+|J\cup Q|-1+2(f_z-f_i-f)+5(d_z-d_i-d)\nonumber\\
    &\ge 2(f_z-f_i-f)+5(d_z-d_i-d)+n_i+|C_i|+|Q|-1\,.\nonumber
  \end{align}
  Given that $G_z$ is a leaf of $\mathcal{B}'$, all its $3(n_z+3)-9=3n_z$
  interior edges are free.
\end{proof}
\begin{restatable}{lemma}{propcountthree}\label{prop:count3}
  Suppose that $G_i$ along with all its children in $\mathcal{B}$ is merged into
  an \textbf{interior node} $G_z$ of $\mathcal{B}'$ using $f$ flips and $d$
  dummy flips.  Then $G_z$ contains at least
  $6(f_z-f_i-f)+15(d_z-d_i-d)+3n_i+3|C_i|+1$ free interior edges.
\end{restatable}
\begin{proof}
  All children of $G_i$ are merged together with $G_i$ into $G_z$. As $G_z$ is
  an interior node of $\mathcal{B}'$, an edge of $T_i$ is flipped in this
  process and the new edge added by this flip is a free interior edge of
  $G_z$. In addition, all edges inside $T_i$ are free interior edges of
  $G_z$. The number of vertices inside $T_i$ is $n_i+\sum_{j\in C_i}n_j$,
  which by Lemma~\ref{p:ni} is at least $n_i+\sum_{j\in
    C_i}(2f_j+5d_j+1)$. Hence the number of free interior edges inside $T_i$
  is at least $3n_i+3|C_i|+\sum_{j\in C_i}(6f_j+15d_j)$.

  From the remaining nodes of $\mathcal{B}$ merged into $G_z$ we get by
  \ref{i:gen} and \ref{i:k4} an additional number of $\sum_{j\in J\setminus
    C_i}(6f_j+15d_j+3)+\sum_{q\in Q^+}(6f_q+15d_q+1)$ free interior edges.
  Summing up yields at least
  \begin{align}
    1&+3n_i+3|C_i|+\sum_{j\in C_i}(6f_j+15d_j)+\sum_{j\in J\setminus
      C_i}(6f_j+15d_j+3)+\sum_{q\in Q^+}(6f_q+15d_q+1)\nonumber\\
    &= 6(f_z-f_i-f)+15(d_z-d_i-d)+3n_i+3|C_i|+1+3|J\setminus C_i|+|Q^+|\nonumber
  \end{align}
  free interior edges in $G_z$.
\end{proof}

\minipar{Case analysis.} We now show that every step of the algorithm {\sc
  4Connect} maintains the invariants \ref{i:root}--\ref{i:k4}. We start with the
case that $\mathcal{T}_i$ forms a checkerboard and then consider the case that
$\mathcal{T}_i$ does not form a checkerboard.
\begin{restatable}{lemma}{propcountfour}\label{prop:count4}
  Suppose that $\mathcal{T}_i$ is a \textbf{checkerboard}. Then $G_z$ fulfills
  invariants \ref{i:root}--\ref{i:k4}.
\end{restatable}
\begin{proof}
  In this case, the algorithm performs $d=1$ dummy flip and $f=|C_i|-3$
  flips. As $\mathcal{T}_i$ is a checkerboard, we have $f_i=d_i=0$ (any previous
  flip in $G_i$ would have created a free interior edge). Recall that $G_i$ has
  $2(n_i+1)$ faces, one of which is the outer face, and hence either
  $|C_i|=n_i$ or $|C_i|=n_i+1$ (see also Lemma~\ref{lem:predummy}). We
  distinguish these two cases.

  \minipar{Case~1: $|C_i|=n_i$.} Then $f=n_i-3$. No edge of $T_i$ is flipped in
  this step, and $G_i$ along with all its children is merged into a leaf $G_z$
  of $\mathcal{B}'$.
  By Lemma~\ref{prop:count2} we find at least
  \[
  6(f_z-(n_i-3))+15(d_z-1)+3n_i+3|C_i|+3|Q|-3= 6f_z+15d_z+3|Q|
  \]
  free interior edges in $G_z$, which noting that $i\in Q$ proves \ref{i:gen}.

  \minipar{Case~2: $|C_i|=n_i+1$.} Then $f=n_i-2$ and $T_i$ is adjacent to three
  distinct triangles from $\mathcal{T}_i$. By Lemma~\ref{lem:predummy}, we have
  $H\in \mathcal{T}_i$, and $H$ is not adjacent to the triangle $F$ selected for
  the dummy flip in this step.  Consequently, the algorithm flips the common
  edge of $H$ and $T_i$.

  If the resulting graph $G_z$ is an interior node of $\mathcal{B}'$, then by
  Lemma~\ref{prop:count3} we find at least
  \[
  6(f_z-(n_i-2))+15(d_z-1)+3n_i+3(n_i+1)+1= 6f_z+15d_z+1
  \]
  free interior edges in $G_z$, which implies \ref{i:k4}.  Otherwise,
  $G_z$ is a leaf of $\mathcal{B}'$ and by Lemma~\ref{prop:count2} we
  find at least
  \[
  6(f_z-(n_i-2))+15(d_z-1)+3n_i+3(n_i+1)+3|Q|-3= 6f_z+15d_z+3(|Q|-1)
  \]
  free interior edges in $G_z$. If $G_z$ is the only vertex of $\mathcal{B}'$,
  then together with the three edges of $T_i$ and noting that $i\in Q$ we obtain
  \ref{i:root} for $G_z$. Otherwise, as an edge of $T_i$ is flipped, also the
  parent $G_p$ of $G_i$ is merged into $G_z$. Therefore $\{i,p\}\subseteq Q$ and
  \ref{i:gen} holds for $G_z$.
\end{proof}

The analysis for the case that $\mathcal{T}_i$ does not form a checkerboard is
split into two lemmata. Lemma~\ref{prop:count1} addresses the case that $G_i$
has two separating triangles that share an edge, whereas Lemma~\ref{prop:count5}
discusses the situation that the triangles in $\mathcal{T}_i$ are pairwise
edge-disjoint.
\begin{restatable}{lemma}{propcountone}\label{prop:count1}
  If $G_i$ contains a doubly trapped edge, then $G_z$ fulfills invariants
  \ref{i:root}--\ref{i:k4}.
\end{restatable}
\begin{proof}
  Let $S$ denote the set of doubly trapped edges in $G_i$, and put $s=|S|$.  As
  $s\ge 1$, we know that $\mathcal{T}_i$ is not a checkerboard and so the
  algorithm flips the edges of an optimal $4$-connector $M$.  As $M$ contains an
  edge of every triangle from $\mathcal{T}_i$, we have $|M|=|C_i|-|M\cap S|$.
  In particular, the choice of the optimal $4$-connector implies $|M\cap
  S|\ge\lceil s/3\rceil$.

  As we flip only edges that are incident to a separating triangle, all free
  interior edges of $G_j$ in $\mathcal{B}$, for $j\in J\cup Q$, remain free
  interior edges of $G_z$ in $\mathcal{B}'$. By \ref{i:gen} and \ref{i:k4} we
  obtain the following lower bound on the number of such edges for $j\in J\cup
  Q^+$ (but not for $j\in Q\setminus Q^+$, a detail that we will get back to
  later):
  \[\sum_{j\in J}(6f_j+15d_j+3)+\sum_{q\in Q^+}(6f_q+15d_q+1)=
  6(f_z-|M|)+15d_z+3|J|+|Q^+|\,.\]
  In addition, all interior edges of $G_i$ that are incident to some triangle in
  $\mathcal{T}_i$ become free in $G_z$ (some of them may have been
  flipped). Every triangle in $\mathcal{T}_i$ has three edges, but some of these
  edges are incident to $T_i$---denote the number of these edges by
  $t\in\{0,1,2,3\}$---or to two triangles of $\mathcal{T}_i$. Therefore, at
  least $3|C_i|-s-t$ interior edges of $G_i$ become free and so there are at
  least
  \begin{eqnarray}
    && 6(f_z-|M|)+15d_z+3|J|+|Q^+|+3|C_i|-s-t \nonumber\\
    &\!\ge\!&
    6f_z+15d_z+3(|J|-|C_i|)+|Q^+|+(6\lceil s/3\rceil-s-t)\label{eq:qp}\\
    &\!\ge\!& 6f_z+15d_z+(6\lceil s/3\rceil-s-t)\label{eq:qp1}
  \end{eqnarray}
  free interior edges in $G_z$, where the first inequality uses
  $|M|=|C_i|-|M\cap S|\le|C_i|-\lceil s/3\rceil$ and the second inequality uses
  $C_i\subseteq J$. If $|M|\le|C_i|-\lceil s/3\rceil-1$, then the last summand
  of \eqref{eq:qp1} becomes $6\lceil s/3\rceil-s-t+6$. Given that $t\le 3$, this
  is at least three and, therefore, the claim follows. Similarly, if $t=0$, then
  $6\lceil s/3\rceil-s-t\ge 3\lceil s/3\rceil\ge 3$, where the last inequality
  is due to $s\ge 1$. Again the claim follows. Hence suppose that $t\ge 1$ and
  $|M|=|C_i|-\lceil s/3\rceil$. We distinguish two cases.

  \textbf{Case 1:} $M$ does not contain an edge of $T_i$.  Let $M_1$, $M_2$, and
  $M_3$ denote the three $4$-connectors that $M$ was selected from. We need to
  show that the last summand in \eqref{eq:qp1} is at least three, for which we
  distinguish three subcases, depending on the residue of $s\!\!\mod 3$.

  If $s\equiv 0\mod 3$, then $|M|=|C_i|-s/3$, that is, $M$ contains exactly
  $s/3$ doubly trapped edges. Every doubly trapped edge appears in exactly one
  of $M_1$, $M_2$, or $M_3$. Therefore $|M|=|M_1|=|M_2|=|M_3|=|C_i|-s/3$.  As
  $t\ge 1$, there is at least one (singly) trapped edge $e$ of $T_i$. Given that
  $e$ is trapped, one of $M_1$, $M_2$, or $M_3$ contains it. Hence, by the
  definition of optimality, also $M$ contains an edge of $T_i$, in contradiction
  to our assumption that it does not.

  If $s\equiv 1\mod 3$, then $6\lceil s/3\rceil-s-t=6(s+2)/3-s-t\ge 5-t$ and so
  the claim holds unless $t=3$. If $t=3$, then all three edges of $T_i$ are
  (singly) trapped. Therefore, each of $M_1$, $M_2$, and $M_3$ and, in
  particular, $M$ contains an edge of $T_i$, in contradiction to our assumption
  that it does not.

  It remains to consider the case $s\equiv 2\mod 3$, which implies $s\ge 2$.
  Then $6\lceil s/3\rceil-s-t=6(s+1)/3-s-t\ge s+2-t\ge 4-t$ and so the claim
  holds unless $t\ge 2$. If $t=3$, then argue as in the preceding case and
  arrive at a contradiction. Hence suppose that $t=2$. Suppose without loss of
  generality that $M=M_1$. Given that $M$ does not contain an edge of $T_i$ and
  two edges of $T_i$ are (singly) trapped, both $M_2$ and $M_3$ contain an edge
  of $T_i$. By the optimality criteria it follows that
  $|M_2|,|M_3|\ge|M|+1=|C_i|-(s+1)/3+1=|C_i|-(s-2)/3$, that is, neither $M_2$
  nor $M_3$ contains more than $(s-2)/3$ doubly trapped edges. On the other
  hand, $M$ contains exactly $(s+1)/3$ doubly trapped edges, which leaves
  $(2s-1)/3$ doubly trapped edges for $M_2$ and $M_3$. But
  $2(s-2)/3=(2s-4)/3<(2s-1)/3$, a contradiction.

  \textbf{Case 2:} $M$ contains an edge $e$ of $T_i$. The last summand in
  \eqref{eq:qp} is $6\lceil s/3\rceil-s-t\ge 3\lceil s/3\rceil-t\ge 3-t\ge 0$.
  In our accounting from \eqref{eq:qp} none of the edges of $T_i$ is counted as
  an interior free edge of $G_z$. But the edge that $e$ is flipped into is a
  free interior edge of $G_z$. So we can raise our count by one. Therefore, if
  $G_z$ is an interior node of $\mathcal{B}'$, then \ref{i:k4} holds and the
  claim follows.

  It remains to consider the case that $G_z$ is a leaf of $\mathcal{B}'$. If the
  other two edges of $T_i$ (other than $e$) are both free interior edges of
  $G_z$, then the claim follows. Otherwise, at least one edge $g\ne e$ of $T_i$
  is not a free interior edge of $G_z$. As $G_z$ is a leaf of $\mathcal{B}'$
  (and so $G_z$ does not have a separating triangle), $g$ is an edge of the
  outer face $T_z$ of $G_z$. As any two triangles in a triangulation share at
  most one edge, it follows that the third edge of $T_i$ (other than $e$ and
  $g$) is a free interior edge of $G_z$. This increases our count by another
  edge.

  If $Q^+\ne\emptyset$, then the claim follows.  Otherwise, we have
  $Q^+=\emptyset$. In particular, for the parent $G_p$ of $G_i$ in $\mathcal{B}$
  we have $p\in Q\setminus Q^+$ and so none of the interior edges of $G_p$ have
  been counted in \eqref{eq:qp}. As every node has at least three interior edges
  and---$G_z$ being a leaf of $\mathcal{B}'$ all its interior edges are
  free---the claim follows.
\end{proof}

For the case that there are no doubly trapped edges in $G_i$ and some flips or
dummy flips have already been executed in $G_i$, the following lemma provides an
upper bound on $|M|$ using the invariants.
\begin{lemma}\label{prop:mci:m3}
  Let $M$ be an optimal $4$-connector for $G_i$. If the triangles in
  $\mathcal{T}_i$ are pairwise edge-disjoint but $\mathcal{T}_i$ is not a
  checkerboard, then $|M|\leq n_i-2f_i-5d_i$. Equality is possible only if $M$
  contains an edge of $T_i$.
\end{lemma}
\begin{proof}
  First assume that $f_i+d_i>0$. Then by \ref{i:k4}, there are at least
  $6f_i+15d_i+1$ free interior edges in $G_i$.  Therefore, at least one of the
  three perfect dual matchings $D_1$, $D_2$, or $D_3$
  (Theorem~\ref{thm:tait4c}), say, $D_1$ contains at least
  $\lceil(6f_i+15d_i+1)/3\rceil=2f_i+5d_i+1$ free interior edges. As none of
  these edges appears in the corresponding $4$-connector $M_1$, we have
  $|M|\le|M_1|\le(n_i+1)-(2f_i+5d_i+1)=n_i-2f_i-5d_i$. In case of equality,
  $M_1$ results from removing only free interior edges from a perfect dual
  matching $D_1$. In particular, as the edge of $D_1$ incident to $T_i$ is not
  interior, both $M_1$ and---by definition of optimality---$M$ contain an edge
  of $T_i$.

  It remains to consider the case $f_i=d_i=0$. As $\mathcal{T}_i$ is not a
  checkerboard, $G_i$ has at least one free interior edge. This edge appears in
  one of the three perfect dual matchings $D_1$, $D_2$, or $D_3$ and therefore
  $|M|\le n_i=n_i-2f_i-5d_i$. If $|M|=n_i$, then suppose contrary to our claim
  that $M$ does not contain any edge of $T_i$. Then $T_i$ is not adjacent to any
  triangle from $\mathcal{T}_i$. (Otherwise, the corresponding edge $e$ shared
  by $T_i$ and a triangle from $\mathcal{T}_i$ appears in one of the three
  $4$-connectors $M_1$, $M_2$, and $M_3$ that $M$ is selected from. As
  $|M_1|=|M_2|=|M_3|$, our optimality criterion selects $M$ to be the
  $4$-connector that contains $e$.)  Given that $|M|=n_i$, we conclude that each
  of the $3n_i$ interior edges of $G_i$ is incident to some triangle from
  $\mathcal{T}_i$. But then $\mathcal{T}_i$ is a checkerboard, contrary to our
  assumption that it is not.
\end{proof}
\begin{restatable}{lemma}{propcountfive}\label{prop:count5}
  Suppose that the triangles in $\mathcal{T}_i$ are pairwise edge-disjoint but
  $\mathcal{T}_i$ is \textbf{not} a \textbf{checkerboard}. Then $G_z$ fulfills
  invariants \ref{i:root}--\ref{i:k4}.
\end{restatable}
\begin{proof}
  By Lemma~\ref{prop:mci:m3} we have $|M|\leq n_i-2f_i-5d_i$. To conclude the
  analysis we distinguish four cases.

  \minipar{Case~1:} $G_z$ is the only node of $\mathcal{B}'$. Using
  Lemma~\ref{prop:count2} with $f=|M|$ and $d=0$, we find at least
  \[
  6(f_z-f_i-f)+15(d_z-d_i)+3n_i+3|C_i|+3|Q|-3 \ge
  6f_z+15d_z+3|Q|-3
  \]
  free interior edges in $G_z$, where the inequality uses $|C_i|=|M|=f$ and
  $n_i\ge f+2f_i+5d_i$. Given that $i\in Q$ and the three outer edges of $T_z$
  are not interior edges of $G_z$, \ref{i:root} follows.

  \minipar{Case~2:} $G_z$ is an interior vertex of $\mathcal{B}'$.  Then using
  Lemma~\ref{prop:count3} with $f=|M|$ and $d=0$ we find at least
  \begin{align*}
    \;&  6(f_z-f_i-f)+15(d_z-d_i)+3n_i+3|C_i|+1\\
    \ge\;& 6(f_z-f_i)+15(d_z-d_i)+3n_i-3f+1\\
    \ge\;& 6f_z+15d_z+1
  \end{align*}
  free interior edges in $G_z$, where the first inequality uses $|C_i|=|M|=f$
  and the second inequality uses $n_i\ge f+2f_i+5d_i$. This proves \ref{i:k4}.

  \minipar{Case~3:} $G_z$ is a leaf of $\mathcal{B}'$ (but not the only
  node). We distinguish two subcases. If $M$ does not contain an edge of $T_i$,
  then Lemma~\ref{prop:mci:m3} yields $f=|M|\leq n_i-2f_i-5d_i-1$. By
  Lemma~\ref{prop:count2}, we find at least
  \begin{align}
    \;& 6(f_z-f_i-f)+15(d_z-d_i)+3n_i+3|C_i|+3|Q|-3\nonumber\\
    \ge\;& 6f_z-6f_i-3f-3(n_i-2f_i-5d_i-1)+15(d_z-d_i)+3n_i+3f+3|Q|-3\nonumber\\
    =\;& 6f_z+15d_z+3|Q|\nonumber
  \end{align}
  free interior edges in $G_z$, where the inequality uses $f\le n_i-2f_i-5d_i-1$
  and $|C_i|=f$. Since $i\in Q$, we have $|Q|\geq 1$ and \ref{i:gen} follows.

  Otherwise, $M$ contains an edge of $T_i$. Then the parent $G_p$ of $G_i$ in
  $\mathcal{B}$ is merged into $G_z$. By Lemma~\ref{prop:count2} we find at
  least
  \begin{align}
    \;& 6(f_z-f_i-f)+15(d_z-d_i)+3n_i+3|C_i|+3|Q|-3\nonumber\\
    \ge\;& 6f_z-6f_i-3f-3(n_i-2f_i-5d_i)+15(d_z-d_i)+3n_i+3f+3|Q|-3\nonumber\\
    =\;& 6f_z+15d_z+3(|Q|-1)\nonumber
  \end{align}
  free interior edges in $G_z$, where the inequality uses $f\le n_i-2f_i-5d_i$
  and $|C_i|=f$. Since $\{i,p\}\subseteq Q$, we have $|Q|\geq 2$ and \ref{i:gen}
  follows.
\end{proof}

\minipar{Summary.} In all cases we have shown that the resulting $4$-block tree
$\mathcal{B}'$ satisfies our invariants. Thus the resulting $4$-connected graph
$G'$ has $n+d$ vertices and at least $6f+15d+3$ edges, where $f$ and $d$ denote
the number of flip and dummy flip operations, respectively, that were executed
during the algorithm. Being a maximal planar graph, $G'$ contains exactly
$3(n+d)-6$ edges. Therefore, $6f+15d+3\le 3(n+d)-6$ and so $2f+4d\le n-3$, as
required. This completes the proof of Lemma~\ref{lem:charging}.

\subsection{Second Step: Eliminate Dummy Vertices}\label{sec:algstep2}

At this stage we have a $4$-connected planar graph $G'$. By Tutte's Theorem such
a graph is Hamiltonian, so consider some Hamiltonian cycle $H'$ of $G'$. It
remains to argue how $G'$ and $H'$ can be used to obtain a short sequence of
edge flips that transform the original graph $G$ into a Hamiltonian graph
$G''$. The following lemma in combination with Lemma~\ref{lem:charging}
completes the proof for the first part of Theorem~\ref{thm:hamflip} (Hamiltonian
graph through flip sequence).

\begin{restatable}{lemma}{eliminate}\label{lem:eliminate}
  Suppose that $G'$ has been obtained from $G$ using $f$ flips and $d$
  dummy flips. Then $G$ can be transformed into a Hamiltonian maximal planar
  graph using at most $f+2d$ edge flips.
\end{restatable}
\begin{proof}
  While it is obvious how to obtain those edges of
  $\mathrm{E}(G')\setminus\mathrm{E}(G)$ that were created using a
  flip (using exactly this flip), we need to argue a bit more for the
  vertices inserted by dummy flips.

  Consider a dummy vertex $v$ and let $T=abc$ denote the triangle of
  $G$ in which $v$ has been inserted
  (cf. Observation~\ref{obs:dummy}).  Obviously $H'$ uses only two of
  the six edges incident to $v$ along a path $uvw$. Our goal is to
  determine the graph $G''$ along with a Hamiltonian cycle $H''$ in
  $G''$. Depending on the relative position of the two edges $uv$ and
  $vw$ we distinguish three cases.

  \textbf{Case 1:} $uv$ and $vw$ are opposite in the circular order of
  edges incident to $v$ (\figurename~\ref{fig:4}(a)--(b)). Then
  exactly one of $uv$ or $vw$ crosses an edge, say $ab$, of $T$ and flipping
  $ab$ yields an edge $uw$ in $G''$. In $H''$ the edge $uw$ takes the
  role of the path $uvw$ in $H'$.

  \textbf{Case 2:} $uv$ and $vw$ are adjacent in the circular order of
  edges incident to $v$ (\figurename~\ref{fig:4}(c)--(d)). Then no
  flip is needed, because the edge $uw$ is an edge of $G$ already.

  \textbf{Case 3:} $uv$ and $vw$ are at distance two in the circular
  order of edges incident to $v$
  (\figurename~\ref{fig:4}(e)--(f)). Then the edges $uv$ and $vw$ each
  intersect an edge of $T$, say, $uv$ intersects the edge $bc$ and
  $vw$ intersects the edge $ab$ of $T$. Then flipping the edge $ab$ in
  $G$ yields the edge $cw$ and a subsequent flip of $bc$ yields the
  desired edge $uw$.

  In every case at most two flips are needed to simulate the sub-path
  of $H'$ passing through a dummy vertex $v$ using a corresponding
  edge in $G''$. Altogether we obtain a Hamiltonian cycle $H''$ in
  $G''$ corresponding to $H'$ in $G'$.
\end{proof}

\begin{figure}[htbp]
  \centering%
  \subfloat[Case~1: before]{\includegraphics[scale=\figscale]{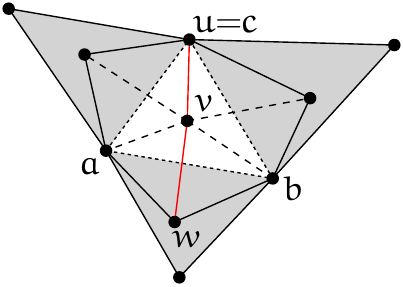}}\hfil
  \subfloat[Case~1: after]{\includegraphics[scale=\figscale]{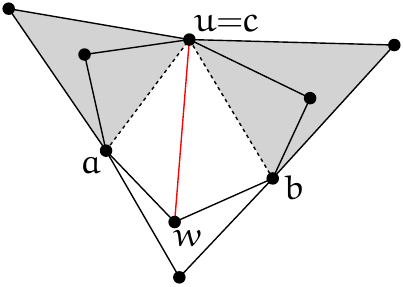}}\hfil\\
  \subfloat[Case~2: before]{\includegraphics[scale=\figscale]{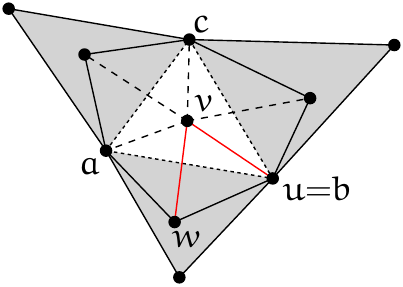}}\hfil
  \subfloat[Case~2: after]{\includegraphics[scale=\figscale]{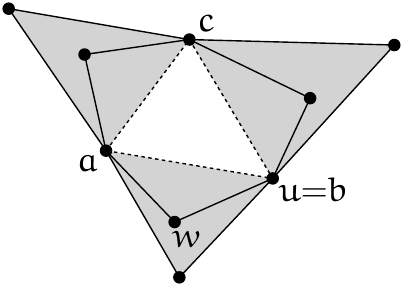}}\hfil\\
  \subfloat[Case~3: before]{\includegraphics[scale=\figscale]{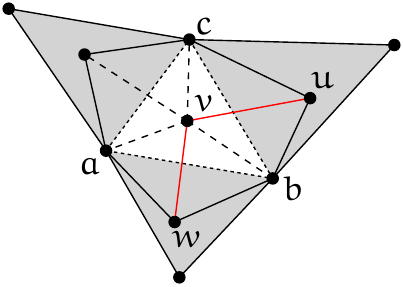}}\hfil
  \subfloat[Case~3: after]{\includegraphics[scale=\figscale]{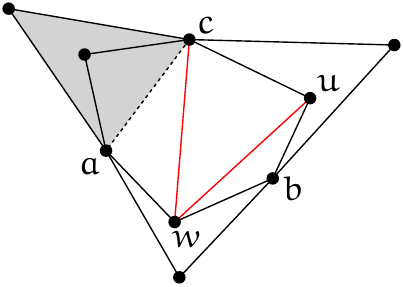}}\hfil
  \caption{Eliminating a dummy vertex $v$ using edge flips. The six edges
    incident to $v$ have been added in $G'$. In this process the three dotted
    edges of the triangle $T$ in the original graph $G$ have been flipped
    away. The three separating triangles adjacent to $T$ are shown shaded. The
    Hamiltonian path $H'$ of $G'$ visits $v$ along the edges shown solid
    red.\label{fig:4}}
\end{figure}

Regarding the second part of Theorem~\ref{thm:hamflip} (subhamiltonian graph
through edge subdivisions) we make a similar argument by translating both flips
and dummy flips into edge subdivisions.

\begin{restatable}{lemma}{subdivide}\label{lem:eliminate2}
  Suppose that $G'$ has been obtained from $G$ using $f$ flips and $d$ dummy
  flips. Then there is a set of at most $f+2d$ edges in $G$ such that subdiving
  them results in a subhamiltonian planar graph.
\end{restatable}
\begin{proof}
  Step~1 of our algorithm identified a set $S$ of $f+3d$ edges in a
  triangulation $G$ such that each separating triangle is incident to at least
  one edge in $S$.  The algorithm destroyed all separating triangles in $G$ by a
  sequence of $f$ flips and $d$ dummy flips to obtain a $4$-connected
  triangulation $G'$.

\begin{figure}[htbp]
  \centering%
  \subfloat[Case~1: before]{\includegraphics[scale=\figscale]{dummy-1}}\hfil
  \subfloat[Case~1: after]{\includegraphics[scale=\figscale]{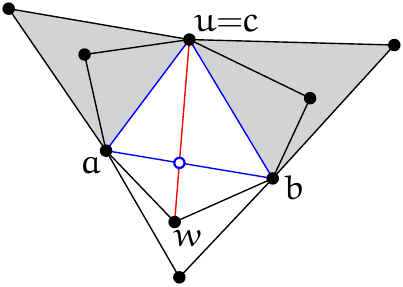}}\hfil\\
  \subfloat[Case~2: before]{\includegraphics[scale=\figscale]{dummy-3}}\hfil
  \subfloat[Case~2: after]{\includegraphics[scale=\figscale]{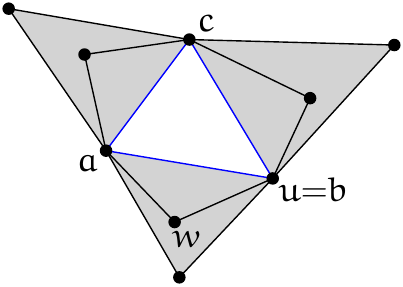}}\hfil\\
  \subfloat[Case~3: before]{\includegraphics[scale=\figscale]{dummy-5}}\hfil
  \subfloat[Case~3: after]{\includegraphics[scale=\figscale]{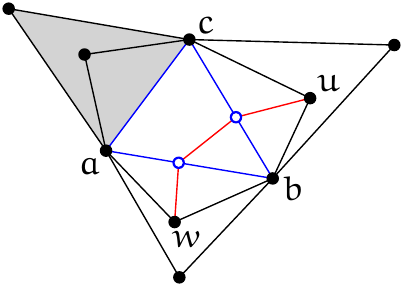}}\hfil
  \caption{Eliminating a dummy vertex $v$ using edge subdivisions. The six edges
    incident to $v$ have been added in $G^+$. In this process the three dotted
    edges of the triangle $T$ in the original graph $G$ have been flipped
    away. The three separating triangles adjacent to $T$ are shown shaded. The
    Hamiltonian path $H^+$ of $G^+$ visits $v$ along the edges shown solid
    red.\label{fig:subdivide}}
\end{figure}

  We now destroy all separating triangles of $G$ by a combination of dummy flips
  and edge subdivisions. Specifically, we construct a graph $G^+$ from $G$ as
  follows.  Perform all dummy flip operations specified by the algorithm;
  instead of each flip operation of the algorithm, subdivide the edge with a new
  vertex; finally insert new edges in the faces incident to subdivision
  vertices: if a face is incident to precisely one subdivision vertex, then
  connect it to the opposite vertex of the face; if it is incident to two or
  three subdivision vertices, then connect all subdivision vertices by an edge
  or a triangle. Since all separating triangles are destroyed and none created,
  the resulting graph $G^+$ is $4$-connected. By Tutte's Theorem $G^+$ contains
  a Hamiltonian cycle $H^+$. It remains to show how to use $G^+$ and $H^+$ to
  replace the $d$ dummy flips with up to $2d$ edge subdivisions.

  Consider a dummy vertex $v$ and let $T=abc$ denote the triangle of $G$ in
  which $v$ has been inserted. By Observation~\ref{obs:dummy} the six triangles
  incident to $v$ right after its insertion remain untouched throughout the
  remainder of the algorithm.  Clearly, $H^+$ uses only two of the six edges
  incident to $v$ along some path $uvw$. We now determine the graph $G^{++}$
  obtained from $G$ with $f+2d$ subdivision vertices along with a Hamiltonian
  cycle $H^{++}$ in $G^{++}$. Similarly to the proof of
  Lemma~\ref{lem:eliminate}, we distinguish three cases depending on the
  relative position of $uv$ and $vw$.

  \textbf{Case 1:} $uv$ and $vw$ are opposite in the circular order of edges
  incident to $v$ (\figurename~\ref{fig:subdivide}(a)--(b)). Then exactly one of
  $uv$ or $vw$ crosses an edge $ab$ of $T$ and subdividing $ab$ yields a new
  path from $u$ to $w$ in $G^{++}$. In $H^{++}$ this path takes the role of the
  path $uvw$ in $H^+$.

  \textbf{Case 2:} $uv$ and $vw$ are adjacent in the circular order of edges
  incident to $v$ (\figurename~\ref{fig:subdivide}(c)--(d)). Then no subdivision
  is needed, because the edge $uw$ is an edge of $G$ already.

  \textbf{Case 3:} $uv$ and $vw$ are at distance two in the circular order of
  edges incident to $v$ (\figurename~\ref{fig:subdivide}(e)--(f)). Then the
  edges $uv$ and $vw$ each intersect an edge of $T$, say, $uv$ intersects the
  edge $bc$ and $vw$ intersects the edge $ab$ of $T$. Subdividing both $ab$ and
  $bc$ yields a new path from $u$ to $w$ in $G^{++}$. In $H^{++}$ this path
  takes the role of the path $uvw$ in $H^+$.

  In every case at most two subdivision vertices are needed to simulate the
  sub-path of $H^+$ passing through a dummy vertex $v$ using a corresponding
  path in $G^{++}$. Consequently, we obtain a Hamiltonian cycle $H^{++}$ in
  $G^{++}$ corresponding to $H^+$ in $G^+$.
\end{proof}

\paragraph{Runtime Analysis.}
It remains to argue that for a maximal planar graph $G$ on $n$ vertices, the
algorithm {\sc 4Connect} can be implemented to run in $O(n^2)$ time.  The
bottleneck is computing an edge-partition into three dual perfect matchings,
which is equivalent to vertex $4$-coloring and can be done in quadratic
time~\cite{bm-gt-2008,rsst-fct-97}.
A $4$-block tree $\mathcal{B}$ of the input $G$ can be computed in $O(n)$
time~\cite{k-amcvr-97}; and we can identify $4$-blocks on the penultimate level
that are checkerboards in $O(n)$ time. Even though the $4$-block tree
$\mathcal{B}$ changes in the course of the algorithm, every new $4$-block is
created by a flip or dummy flip operation, and hence contains a free interior
edge. Checkerboards, however, do not contain free interior edges. Consequently,
every checkerboard $G_i$ that the algorithm encounters is already present in the
input graph at the penultimate level of $\mathcal{B}$. Therefore it is enough to
identify checkerboards during a preprocessing step. It is straightforward to
implement Step~(3) of {\sc 4Connect} in $O(n)$ time for all checkerboards.

Step~(2) considers a node $G_i$ that is not a checkerboard, and requires an
optimal $4$-connector $M$, which in turn depends on the edge partition of $G_i$
into three perfect dual matchings. If we call a $4$-coloring algorithm for each
node $G_i$ in the course of the algorithm, the running time would be
$O(n\cdot n^2)=O(n^3)$. However, it suffices to compute an edge partition into
three perfect dual matchings $D_1\cup D_2\cup D_3$ such that every separating
triangle is incident to precisely one edge from each matching
(Corollary~\ref{c:dualhittriangle}) at preprocessing in $O(n^2)$ time. Since
separating triangles are destroyed but never created by the algorithm, each
separating triangle is already present in the input graph $G$, and one edge is
contained in each of $D_1$, $D_2$, and $D_3$. Consequently, for each $4$-block
$G_i$ we can use the restrictions of the initial edge-partition
$D_1\cup D_2\cup D_3$ to the edges of $G_i$ rather than recomputing a new
edge-partition for $G_i$.  Finally, the maintenance of the $4$-block tree
$\mathcal{B}$ takes $O(n)$ time overall, since every flip and dummy flip merges
two nodes of $\mathcal{B}$ and the initial number of nodes is linear.

Our algorithm {\sc 4Connect} returns a $4$-connected graph $G'$ with $O(n)$
vertices. A Hamiltonian cycle $H'$ in $G'$ can be computed in $O(n)$ time and
space~\cite{cn-hcpls-89}. The local operations in Lemmata~\ref{lem:eliminate}
and~\ref{lem:eliminate2} that replace a dummy flip by ordinary flips or
subdivisions, respectively, can be implemented in $O(n)$ time.

\subsection{Lower bound}

The best lower bound we know can be obtained using the following standard
construction~\cite{ahk-tpst-08,g-cp-03}.
\hamlower* %
\begin{proof}
  Consider an arbitrary maximal planar graph $G_0=(V_0,E_0)$ on $t$ vertices, for
  $t>4$. Let $G=(V_0\cup V_1,E)$ be the graph obtained by inserting a vertex in
  every face of $G_0$ and connecting it to the vertices bounding the face. $G$
  has $n:=t+2t-4=3t-4$ vertices.

  Suppose that there exists a sequence of $k$ flips that transforms $G$ into a
  Hamiltonian graph $H$. Let $C$ be a Hamiltonian cycle in $H$. Since $|V_0|=t$
  and $|V_1|=2t-4$, there must be at least $2t-4-t=t-4$ vertices in $V_1$ that
  are followed by another vertex of $V_1$ in $C$. In $G$, however, $V_1$ forms
  an independent set. Note that every flip in $G$ creates at most one edge
  between vertices of $V_1$.  Hence, $k\geq t-4=(n+4)/3-4=(n-8)/3$.

  Analogously, suppose that we can obtain a subhamiltonian graph $H_0=(V_0\cup
  V_1\cup V_2, E')$ by performing $|V_2|=k$ edge subdivisions in $G$. Augment
  $H_0$ to a Hamiltonian graph $H$ and let $C$ be a Hamiltonian cycle in $H$.
  Since $|V_0|=t$, $|V_1|=2t-4$, and $|V_2|=k$, there must be at least
  $2t-4-t-k=t-k-4$ vertices in $V_1$ that are followed by another vertex of
  $V_1$ in $C$. In $G$ and $H$, however, $V_1$ forms an independent set. Hence,
  $k\geq t-4=(n+4)/3-4=(n-8)/3$.
\end{proof}

\section{Conclusions}

We conclude by listing some open problems. An obvious problem is to obtain a
tight bound for the maximum number of flips needed to transform any
triangulation on $n$ vertices into a Hamiltonian triangulation.
We proved an upper bound of roughly $n/2$ (Theorem~\ref{thm:hamflip})
and a lower bound of roughly $n/3$ (Theorem~\ref{thm:hamlower}).

We have shown that sometimes fewer flips suffice to reach a Hamiltonian
triangulation than a $4$-connected one. Is there a similar difference for
simultaneous flips?  We gave an almost tight bound of roughly $2n/3$ edges to
obtain a $4$-connected triangulation (Theorem~\ref{thm:upsim4}).  Can we get a
better bound on the number of edges in a simultaneous flip when we ask for a
Hamiltonian triangulation only? The lower bound to reach a Hamiltonian
triangulation is the same as in the case of flip sequences: $n/3$. Again, a gap
remains.

Another obvious task is to improve the upper bound for the diameter of the
combinatorial flip graph. In the current upper bound, $4n$ of $5n$ flips are
spent transforming Hamiltonian triangulations to the canonical one. Perhaps,
this is the most promising place to look for improvements. As for lower bounds,
a simple degree argument gives a lower bound of roughly $2n$
flips~\cite{k-dfts-96}. Recently, Frati~\cite{f-lbdfg-15} improved this bound to
roughly $7n/3$ flips.

Finally, there is the algorithmic question of how to efficiently compute a flip
sequence between two triangulations that meets the diameter bounds. We gave a
quadratic time algorithm whose bottleneck is the computation of a Tait partition
(equivalently, a 4-coloring). One perfect dual matching can be computed in
linear time~\cite{bbdl-eapmt-01}, but our averaging argument needs an edge
partition into perfect dual matchings. Is there a different way to do the
accounting that leads to a subquadratic algorithm?

A more ambitious goal is to settle the complexity of the flip distance
problem in the combinatorial setting: Given a positive integer $k$ and two
triangulations $T_1$ and $T_2$ on $n$ vertices, is there a sequence of at most $k$
combinatorial flips that transforms $T_1$ into $T_2$? In the geometric setting
this problem is known to be APX-hard~\cite{p-fdtppsa-14}.

\paragraph{Acknowledgements} This work began at the 12th Gremo's Workshop on
Open Problems (GWOP), June 30--July 4, 2014, in Val Sinestra (GR),
Switzerland. We thank all participants for the productive and positive
atmosphere, and in particular Radoslav Fulek, Anna Gundert, Malte Milatz,
Bettina Speckmann, Sebastian Stich, and Tibor Szab{\'o} for inspiring
discussions.

\bibliographystyle{mh-url} %
\bibliography{arcs}

\end{document}